\documentclass[10pt,journal,twocolumn,final]{IEEEtran}
\let\labelindent\relax
\usepackage{enumitem}
\setlist[itemize]{leftmargin=*}
\usepackage[ruled,vlined]{algorithm2e}
\SetArgSty{textup}
\SetKw{inputs}{Inputs:}
\SetKw{initialization}{Initialization:}
\SetKw{declareopt}{Declaring optimization variables:}
\SetKw{declarepoly}{Constructing polynomials:}
\SetKw{declareconst}{Declaring affine constraints:}
\SetKw{solving}{Solving:}
\SetKw{outputs}{Return outputs:}
\SetKw{controlgains}{Calculating control gains:}
\SetKw{observergains}{Calculating observer gains:}
\usepackage{graphics} 
\usepackage{graphicx}
\usepackage{epsfig} 
\usepackage{float}
\usepackage{times} 
\usepackage{amsthm} 
\usepackage{stdmath}
\usepackage{amsmath}  
\usepackage{mathtools}
\usepackage{amsfonts}
\newtheorem{theorem}{Theorem}

\newtheorem{lemma}{Lemma}

\allowdisplaybreaks
\makeatletter

\newcommand{\Rmnum}[1]{\expandafter\@slowromancap\romannumeral #1@}
\makeatother
\usepackage{graphicx}
\usepackage{caption}
\usepackage{subfigure}
\makeatother

\newcommand{\pfx}{\frac{\partial}{\partial x}}
\newcommand{\pfxi}{\frac{\partial}{\partial \xi}}

\newcommand{\igzo}{\int_0^1}
\newcommand{\igzx}{\int_0^x}
\newcommand{\igxo}{\int_x^1}

\newcommand{\wh}{\hat{w}}

\newcommand{\lt}{L_2(0,1)}

\newcommand{\mcl}[1]{\mathcal{#1}}

\newcommand{\hlf}{\frac{1}{2}}
\newcommand{\pop}{\mathcal{P}}
\newcommand{\pinv}{\mathcal{P}^{-1}}

\usepackage{tikz, xcolor}
\usetikzlibrary{shapes,arrows}
\usetikzlibrary{positioning}
\usetikzlibrary{calc}
\usetikzlibrary{fit}
\tikzset{
  -|-/.style={
    to path={
      (\tikztostart) -| ($(\tikztostart)!#1!(\tikztotarget)$) |- (\tikztotarget)
      \tikztonodes
    }
  },
  -|-/.default=0.5,
  |-|/.style={
    to path={
      (\tikztostart) |- ($(\tikztostart)!#1!(\tikztotarget)$) -| (\tikztotarget)
      \tikztonodes
    }
  },
  |-|/.default=0.5,
}

\tikzstyle{decision} = [diamond, draw, text width=4.5em,
                        text badly centered, node distance=2cm,
                        inner sep=0pt]
\tikzstyle{block} = [rectangle, draw, text width=10cm,
                     text centered,
                     minimum height=1cm, node distance=3cm]

\tikzstyle{line} = [draw, -latex']
\tikzstyle{smallblock} = [rectangle, draw,
                     text centered,
                     minimum height=1cm, node distance=3cm]
\tikzstyle{cloud} = [draw, circle, node distance=2.5cm, minimum height=2.8em]
\tikzstyle{blank} = [node distance=1cm]

\begin{document}
%
\title{A Convex Sum-of-Squares Approach to Analysis, \\
 State Feedback and Output Feedback \\
  Control of Parabolic PDEs}
%
%
%

\author{Aditya~Gahlawat
         and~Matthew.~M.~Peet
\thanks{ This research was supported by the Chateaubriand program and NSF CAREER Grant CMMI-1151018. }         
\thanks{Aditya Gahlawat is with the Department
of Mechanical, Materials and Aerospace Engineering at the Illinois Institute of Technology, Chicago,
IL, 60616 USA e-mail: (agahlawa@hawk.iit.edu).}
\thanks{Matthew. M. Peet is with the School of Engineering of Matter, Transport and Energy at Arizona State University, Tempe, AZ, 85287-6106 USA e-mail: (mpeet@asu.edu).}
\thanks{Color versions of one or more of the figures in this paper are available online at http://ieeexplore.ieee.org.}
\thanks{Digital Object Identifier 10.1109/TAC.2016.2593638}
}


\maketitle
\begin{abstract}
We present an optimization-based framework for analysis and control of linear parabolic Partial Differential Equations (PDEs) with spatially varying coefficients without discretization or numerical approximation. For controller synthesis, we consider both full-state feedback and point observation (output feedback). The input occurs at the boundary (point actuation). We use positive definite matrices to parameterize positive Lyapunov functions and polynomials to parameterize controller and observer gains. We use duality and an invertible state variable transformation to convexify the controller synthesis problem. Finally, we combine our synthesis condition with the Luenberger observer framework to express the output feedback controller synthesis problem as a set of LMI/SDP constraints. We perform an extensive set of numerical experiments to demonstrate accuracy of the conditions and to prove necessity of the Lyapunov structures chosen. We provide numerical and analytical comparisons with alternative approaches to control including Sturm Liouville theory and backstepping. Finally we use numerical tests to show that the method retains its accuracy for alternative boundary conditions.

\end{abstract}
\IEEEpeerreviewmaketitle

\section{Introduction}
Partial Differential Equations (PDEs) are used to model quantities which vary in both space and time with early examples including the D'Alembert wave equation (1746); the Euler-Bernoulli beam (1750); the Euler equations (1757); and the Fourier heat equation (1822). Today, the use of PDE models has expanded to include phenomena such as the magnetohydrodynamics of plasma in a fusion reactor~\cite{witrant2007control}, tumour growth, infectious diseases, and ecological succession~\cite[Chapter~$11$]{murray2002mathematical}. However, despite the variety of phenomena modeled by PDEs, compared to the literature on Ordinary Differential Equations (ODEs), our knowledge of how to analyze and control PDEs remains incomplete.

Consider the following class of scalar-valued anisotropic parabolic PDEs with input $u(t) \in \R$,
\begin{equation}
 w_t(x,t)=a(x)w_{xx}(x,t)+b(x)w_x(x,t)+c(x)w(x,t), \label{eqn:prob:PDE_form}
\end{equation} $x \in [0,1]$, $t \ge 0$,
which has output $v(t)=w(1,t) \in \R$ and mixed boundary conditions of the form
\begin{equation}
\label{eqn:prob:PDE_form_BC}
w(0,t)=0, \qquad w_x(1,t)=u(t),
\end{equation}
where $a$, $b$ and $c$ are polynomials with $a(x) \geq \alpha >0$, for $x \in [0,1]$. We assume the controller is parameterized by scalar $R_1$ and function $R_2$ as $u(t)=R_1 \hat w(1,t)+\int_0^1 R_2(x)\hat w(x,t)dx$ where $\hat w$ is an estimate of $w$ obtained from some set of observer dynamics. The objective of the paper is to propose an optimization-based method for determining controller gains $R_1$ and $R_2$ and observer dynamics which minimize certain closed-loop gains.

Control of PDE models is a challenging problem in that slight variations in the type of PDE, boundary conditions, etc. may dramatically alter properties of the solution~\cite{lions1972non}. The model defined above is classified as an anisotropic parabolic PDE with point inputs and point outputs. The term anisotropic means that the values of the coefficients $a(x),b(x)$ and $c(x)$ depend on the spatial variable $x\in [0,1]$. Examples of anisotropic systems include heat conduction with non-homogeneous conductive properties or a wave propagating through a medium of varying density. The term point input (boundary actuated) means that the control input determines one of the boundary values and therefore has no direct measurable effect on Equation~\eqref{eqn:prob:PDE_form}. This is in contrast to the case of distributed inputs, wherein the control effort is spread over some measurable subset of the domain. In a similar manner, the term point output means that the sensor measures the state at a single point in the domain and hence the output operator is unbounded in the $L_2$ induced norm.

Perhaps the most common approach to analysis and control of PDEs is based on the use of discrete approximation. Such approximation techniques typically use a model reduction wherein the PDE is approximated by a set of ODEs. Finite-dimensional linear control theory is then used to analyze stability and design control laws for the finite-dimensional approximations~\cite{morris1994design,morris2010approximation}. Furthermore, results have been obtained which show that as the order of the discrete approximation increases, stability of the closed-loop approximations will eventually imply stability of the closed-loop PDE. A disadvantage of the discrete approach, however, is that the required order of the approximation cannot be established a priori. Consequently, the stability of any particular approximation is not guaranteed to imply stability of the actual PDE. For this reason, among others, there has been some interest in finding approaches to analysis and control which can be applied directly to the PDE model without the use of discretization or numerical approximation. Such methods are sometimes termed direct or infinite-dimensional.

There has been significant progress in the use of direct methods for control of PDE systems. One approach is to express the control problem as the solution to a set of operator-valued Riccati equations. This approach was applied to distributed input/distributed output optimal control problems in~\cite{van1993h}. The problem of point actuation with full-state feedback was considered in~\cite{lasiecka2000control} (and related work) and extended in~\cite{lasiecka1994control} to output feedback controller synthesis through the use of a Luenberger observer. An alternative Riccati-based approach for static output feedback of a certain class of well-posed operators can be found in~\cite{staffans1997quadratic,staffans1998quadratic,weiss1997optimal}. A limitation of these Riccati-based methods, however, is that they rely on finite-dimensional numerical methods for obtaining the operator-valued solution. While convergence of these approximations has been demonstrated~\cite{lasiecka2000control}, for a given level of approximation, it is not possible to determine whether existence of a solution implies the closed loop is stable when applied to the original PDE.

Backstepping~\cite{krstic2008boundary} is a  popular and well-developed method for boundary control of parabolic PDE systems. This approach is based on the use of a boundary controller to transform the PDE to a simpler model for which the existence of a decreasing Lyapunov function has previously been established. The backstepping approach is commonly used in the literature and has been extended to many classes of PDE systems - see, e.g.~\cite{krstic2008adaptive,smyshlyaev2007adaptive,smyshlyaev2007adaptive2,smyshlyaev2006lyapunov}. A highlight of the backstepping method is that for certain types of system, stabilizability guarantees the existence of a backstepping transformation. However, a drawback of the backstepping approach is that it is not based on optimization, but rather typically requires numerical integration of a PDE in order to obtain the stabilizing controller - thereby making extensions to robust and optimal control more difficult.  Although a complete survey the of the literature on direct control of PDEs is beyond the scope of this paper, we do note some other significant results on the use of Lyapunov functions for analysis and control of infinite dimensional systems including: a rotating beam~\cite{coron1998stabilization}; quasilinear hyperbolic systems~\cite{coron2008dissipative}; and control of systems governed by conservation laws \cite{coron2007strict}.
As an alternative to Lyapunov-based methods, a classical spectral approach to stability and stabilization is based on Sturm-Liouville theory. In particular, the differential operators which define the PDEs in this paper can be adapted to the Sturm-Liouville framework, from whence one can attempt to determine stability and design \textbf{static} output-feedback controllers. As is demonstrated in Section~\ref{sec:comparison}, however, the use of dynamic output feedback offers considerable advantages over this classical framework.

The goal of this paper is to design stabilizing static state feedback and dynamic output feedback controllers for PDE systems. Our approach is inspired by the use of Linear Matrix Inequalities (LMIs) and Semi-Definite Programming (SDP) in control of ODEs. For stability analysis, as discussed in Sections~\ref{sec:posop} and~\ref{sec:stability}, we use positive definite matrices to create a linear parametrization of a cone of Lyapunov functions which are positive on the Hilbert space $L_2$. Specifically, the Lyapunov functions have the quadratic form $V=\ip{\mcl{Z}(w)}{P\mcl{Z}(w)}_{L_2}$ where $w\in L_2$ is the infinite-dimensional state, $P$ is a positive definite matrix and $\mcl{Z}$ is a fixed vector of multiplication and integral operators with monomial multipliers and kernels. The derivative of the Lyapunov function is likewise constrained to be a negative definite quadratic form. If such a Lyapunov function exists it directly proves stability of the PDE - i.e. there is no numerical approximation. For state-feedback controller synthesis, the controller, as defined above, is parameterized by a scalar $R_1$ and a function $R_2$. Combining these gains with the quadratic Lyapunov functions used for stability analysis yields synthesis conditions which are bilinear in the design variables. However, as described in Sections~\ref{sec:operators} and~\ref{sec:synthesis}, by defining an invertible state transformation and a variable substitution, we derive synthesis conditions which are linear in the optimization variables. Next, in Section~\ref{sec:obsynth} we introduce a class of infinite-dimensional Luenberger observers with observer gains, again parameterized by the coefficients of polynomials. Again, using the Lyapunov function from Section~\ref{sec:posop} and the invertible state variable transformation from Section~\ref{sec:operators}, we obtain SDP-based observer synthesis conditions. Finally, in Section~\ref{sec:num_results}, we verify the accuracy of the method with a series of numerical tests which indicate that the proposed stability conditions are accurate to several decimal places and suggest that for any suitably controllable and observable system, the algorithm will return an observer-based controller. This is followed by Section~\ref{sec:comparison}, wherein we include numerical and analytical comparisons with other results in the literature, including Sturm-Liouville and backstepping.

A significant contribution of the paper, in addition to a new approach to analysis and control of PDEs, lies in the flexibility of the optimization-based approach. Specifically, as the use of LMIs for control of ODEs enabled the field of robust control, so too does our LMI/Lyapunov-based approach to control of PDEs allow the extension to analysis and control of PDEs with parametric uncertainty, PDEs with nonlinearity, multivariate PDEs and PDEs coupled with ODEs or delays. Finally, we note that our approach is complementary to several recent results in the use of LMIs for stability and control of PDEs, including, e.g. our early work in~\cite{papachristodoulou2006analysis}, modeling and control of nonlinear dynamic systems in~\cite{tanaka2009sum}, stability analysis of semilinear parabolic and hyperbolic systems in~\cite{fridman2009lmi} and the numerous results contained in~\cite{orlov2014advanced}.


\section{Notation}\label{sec:notation}
We denote the vector space of $m$-by-$n$ real matrices by $\R^{m \times n}$ and the subspace of symmetric matrices by $\S^n \subset \R^{n \times n}$ where the multiplicative and additive identities are denoted by $I_n \in \S^n$ and $0_{m,n} \in \R^{m \times n}$, respectively. For $P \in \S^n$, $P>0$ $(P \geq 0)$ denotes that $P$ is a positive definite (positive semi-definite) matrix. The spaces of $n-$times continuously differentiable and infinitely differentiable functions on an interval $W \subset \R$ are denoted by $C^n(W)$ and $C^\infty(W)$, respectively. In a similar manner, $C^{n,m}(W_1,W_2)$ represents the space of $n$ and $m-$times continuously differentiable functions on intervals $W_1 \subset \R$ and $W_2 \subset \R$, respectively. The shorthand $u_x$ and $u_t$ denote the partial derivative of $u$ with respect to independent variables $x$ and $t$, respectively. For a bivariate function, $f(x,y)$, we denote $D_1 f:=f_x$ and $D_2 f:=f_y$ - i.e. $D_1$ is differentiation with respect to the first variable and $D_2$ is differentiation with respect to the second. In a similar manner, $D_1^2:=f_{xx}$ and $D_2^2:=f_{yy}$. Recall $L_2(W)$ is the standard Hilbert space of square Lebesgue integrable functions with standard norm and inner product. We use $H^n(W)$ to denote the Sobolev subspace $H^n(W):=\left\{y \in L_2(W) \,: \, \frac{d^{n}y}{dt^{n}} \in L_2(W)\right\}$ with inner product $\ip{x}{y}_{H^n} = \sum_{m=0}^n\ip{\frac{d^m x}{dt^m}}{\frac{d^m y}{dt^m}}_{L_2}$. We occasionally let $L_2(0,1):=L_2([0,1])$ and $H^n(0,1):=H^n([0,1])$. For normed spaces $X$ and $Y$, $\mcl{L}(X,Y)$ denotes the Banach space of bounded linear operators from $X$ to $Y$ with induced norm $\|\mcl{G}\|_\mathcal{L}:=\sup_{\norm{x}_X=1}\norm{\mcl{G}x}_Y$ and we denote $\mcl{L}(X):=\mcl{L}(X,X)$. We define $Z_d(x) \in \R^{d+1 \times 1}$ to be the column vector of all monomials in variables $x$ of degree $d$ or less arranged in increasing lexicographical order. We often use the notation $Z_d(x,y):=Z_d([x;y])$ to denote the vector of monomials in both $x$ and $y$. For any function $T\in L_2$ we use $\mathcal{M}_T: L_2 \rightarrow L_2$ to denote the multiplier operator defined by $T$. i.e. $(\mathcal{M}_T w)(x)=T(x)w(x)$. For any functions $M,K_1,K_2 \in C^\infty$ we define
\begin{align}
&\left(\mcl{X}_{\{M,K_1,K_2\}}w  \right)(x) \notag \\
&\label{eqn:X}=M(x)w(x)+\igzx K_1(x,\xi)w(\xi)d\xi+\igxo K_2(x,\xi)w(\xi)d\xi.
\end{align}

\section{Problem Statement}\label{proset}
 For the system of Equations~\eqref{eqn:prob:PDE_form}~-~\eqref{eqn:prob:PDE_form_BC}, the strict positivity of $a(x)$ implies that the differential operator defining the PDE is uniformly elliptic~\cite[Section~$6.1$]{evans2009partial}. This means that $w$ diffuses from higher density to lower density, a property which is representative of most physical systems.  The choice of sensor and actuator location is somewhat arbitrary. For the heat equation, input $w_x(1,t)=u(t)$ would represent heat flow into the rod and the output $v(t)=w(1,t)$ represents the temperature of the rod at that point. Note that the results of this paper can be adapted to Dirichlet, Neuman and Robin boundary conditions with only slight modifications to the conditions and proofs. These extensions are addressed in Section~\ref{sec:ABC}.

The goal of this article is to design algorithms which resolve the following problems:
\begin{enumerate}
\item \textit{Stability Analysis}: Establish global exponential stability of the trivial solution $w \equiv 0$ of the autonomous system $u(t)=0$ and determine the exponential rate of decay $\delta$.
\item \textit{State feedback control:} If the autonomous system is unstable, construct gains $R_1 \in \R$ and $R_2(x) \in C^\infty(0,1)$ such that if
\begin{equation}
u(t)=R_1 w(1,t)+\igzo R_2(x)w(x,t)dx,\label{eqn:feedback_structure}
\end{equation}
then the trivial solution $w \equiv 0$ is globally exponentially stable with some desired rate of decay, $\mu$.
\item \textit{Output feedback control:} If only output feedback is available ($v(t)=w(1,t)$), construct gains $L_1\in C^\infty(0,1)$ and $L_2 \in \R$ such that for stabilizing gains $R_1$ and $R_2$, if
\[
u(t)=R_1 \wh(1,t)+\igzo R_2(x)\wh(x,t)dx,
\]
where $\hat w$ satisfies
 \begin{align}
  \wh_t(x,t)=&a(x)\wh_{xx}(x,t)+b(x)\wh_x(x,t) \notag \\
  &\label{eqn:prob:observer_form}\qquad \quad +c(x)\wh(x,t)+L_1(x)\left(\hat{v}(t)-v(t) \right),
\end{align}
for $v(t)=w(1,t)$ and $\hat v(t)=\hat w(1,t)$  with boundary conditions
\begin{equation}
\label{eqn:prob:observer_form_BC}
 \wh(0,t)=0, \qquad \wh_x(1,t)=u(t)+L_2\left(\hat{v}(t)-v(t) \right),
\end{equation}
then the trivial solution $w\equiv 0$ of Equations~\eqref{eqn:prob:PDE_form}~-~\eqref{eqn:prob:PDE_form_BC} is globally exponentially stable.
\end{enumerate}

Note that if we consider only bounded linear operators, then the structure of the controller in~\eqref{eqn:feedback_structure} is not restrictive, as any bounded linear functional can be represented in this way using only the integral form (second term). However, we also would like to consider unbounded operators and hence we include the term $R_1 w(1,t)$ as well. If controllers of this form prove inadequate, then one can generalize the structure further to include terms such as $\int_0^1 R_3(x)w_x(x,t)dx$ as in~\cite{gahlawat2011designing}.

The choice for the structure of the Luenberger observer was similarly determined in an ad-hoc manner through inclusion of terms necessary to achieve separation of controller synthesis and observer design objectives. That is, the goal of the observer is to stabilize the dynamics of the estimation error $e=\wh-w$ and the terms in Equations~\eqref{eqn:prob:observer_form}~-~\eqref{eqn:prob:observer_form_BC} were chosen as the minimal necessary to achieve this objective. Again, this structure mirrors the structure of observers found in the backstepping approach.

\subsection{Existence and Uniqueness}
We now briefly discuss the uniqueness and existence of solutions. Define the operator
\begin{equation}\label{eqn:exist:A}
\mcl{A}=a(x) \frac{d^2}{dx^2}+b(x)\frac{d}{dx} + c(x).
\end{equation} It is known that the operator $\mcl{A}$ restricted to space
\begin{equation}\label{eqn:D0}
\mcl{D}_0=\{w \in H^2(0,1): \quad w(0)=w_x(1)=0\},
\end{equation}
 generates a strongly-continuous semigroup, or a $C_0$-semigroup, on $L_2(0,1)$ (see, e.g.,~\cite[Section~$2.1$]{curtain1995introduction}). More precisely, one can represent $\mcl{A}$ as the negative of a Sturm-Liouville operator on $\mcl{D}_0$ and hence, using the spectral properties of a Sturm-Liouville operator, it can be proven that $\mcl{A}$ restricted to $\mcl{D}_0$ generates a $C_0$-semigroup on $\lt$~\cite{delattre2003sturm}. Thus, using Theorems $3.1.3$ and $3.1.7$ in~\cite{curtain1995introduction} we conclude that in the autonomous case ($u(t)=0$), for any initial condition $w_0 \in \mcl{D}_0$ there exists a unique classical solution of~\eqref{eqn:prob:PDE_form}~-~\eqref{eqn:prob:PDE_form_BC}.

For the state-feedback case, using a fixed point argument similar to the one presented in~\cite{balogh2004stability} it can be shown that for $R_1 \in \R$ and $R_2 \in L_\infty (0,1)$, the closed loop system~\eqref{eqn:prob:PDE_form}~-~\eqref{eqn:prob:PDE_form_BC} with
\[
u(t)=R_1 w(1,t)+\igzo R_2(x)w(x,t)dx,
\]
admits a unique local in time solution $w \in C^{1,2}((0,T),[0,1])$, for $T>0$ sufficiently small, for any initial condition $w_0 \in \mcl{D}$, where
  \begin{align}
\mcl{D}=\{&w \in H^2(0,1) \colon w(0)=0 \text{ and } \notag \\
&\label{eqn:D} \qquad \qquad w_x(1)=R_1 w(1)+\igzo R_2(x)w(x)dx\}.
\end{align} Thus if we can establish that any solution of the closed loop system decays exponentially, then this implies the existence and uniqueness of a unique classical solution $w \in C^{1,2}((0,\infty),[0,1])$ for any $w_0 \in \mcl{D}$. The proof of this statement has been omitted, but follows the arguments presented in~\cite[Section~6]{balogh2004stability}.

Finally, consider the observer-based controller as defined in Equations~\eqref{eqn:prob:PDE_form}~-~\eqref{eqn:prob:PDE_form_BC} and~\eqref{eqn:prob:observer_form}~-~\eqref{eqn:prob:observer_form_BC}. Define the estimator error as $e=\wh-w$, which is governed by
\begin{equation}
 e_t(x,t)=a(x)e_{xx}(x,t)+b(x)e_x(x,t)+c(x)e(x,t) +L_1(x)e(1,t),\label{eqn:exist:error_form}
\end{equation}
with boundary conditions
\begin{equation}
\label{eqn:exist:error_form_BC}
 e(0,t)=0, \qquad e_x(1,t)=L_2e(1,t).
\end{equation}
It has been established in~\cite[Section~2]{fridman2009lmi} that for $L_1 \in C^1(0,\infty)$ and $L_2 \in \R$, Equations~\eqref{eqn:exist:error_form}~-~\eqref{eqn:exist:error_form_BC}, if exponentially stable,  admit a unique local in time solution $e \in C^{1,2}((0,T),[0,1])$, for $T>0$ sufficiently small, for any initial condition $e_0 \in \mcl{D}_e$, where
\begin{equation}\label{eqn:D_e}
\mcl{D}_e=\{w \in H^2(0,1) \colon w(0)=0 \text{ and }w_x(1)=L_2 w(1)\}.
\end{equation} \
Therefore, if we can establish that any solution of the coupled closed-loop dynamics decays exponentially, then the local in time solution can be extended to a classical solution $e \in C^{1,2}((0,\infty),[0,1])$ for any initial condition $e_0 \in \mcl{D}_e$.


\section{A Framework for Stability Analysis and Control}\label{sec:framework}
Our approach is motivated by the use of LMIs for optimal control of finite-dimensional systems. For example, consider the autonomous finite-dimensional ODE
 \[
 \dot{x}(t)=Ax(t),
\]
 where $x(t) \in \R^n$. This ODE is exponentially stable if and only if there exists a positive definite matrix $P \in \S^n$ such that
\[
A^T P+ P A < 0.
\]
Feasibility of this LMI implies that the Lyapunov function $V(x)=x^T P x$ is positive definite and its derivative along solutions $\dot V(x)=x^T (A^T P+PA)x$ is negative definite. For stability of PDEs, our approach is to use positive matrices to define positive quadratic Lyapunov functions, except that instead of $V(x)=x^T P x$, we will use the form $V(w)=\ip{\mathcal{Z}(w)}{P \mathcal{Z}(w)}$, where $\mathcal{Z}:L_2 \rightarrow \R^p$ is a vector of bases for a subspace of linear operators on $L_2$ (similar to how $x=[x_1,\cdots,x_n]^T$ is a vector of bases for the space of linear functions on $\R^n$). In our case, however, $\mathcal{Z}$ parameterizes a subspace of multiplier and integral operators with polynomial multipliers and semi-separable kernels. Then, if $P>\epsilon I$, it has a symmetric square root and hence $V(w)=\ip{P^{\half}\mathcal{Z}(w)}{P^{\half} \mathcal{Z}(w)}\geq \epsilon \norm{w}^2$. 
For the time derivative, we will similarly require $\dot V(w(t))+\mu V(w(t))=-\ip{\mathcal{Z}(w(t))}{Q \mathcal{Z}(w(t))}$, for some scalar $\mu>0$ and $Q > 0$ where here and throughout the paper we denote by $\dot V$ the function which satisfies $\dot V(w(t)):=\frac{d}{dt} V(w(t))$ for any solution of the associated PDE - i.e. the derivative along solutions or time-derivative. Existence of such $P,Q>0$ implies exponential stability of the system. As was done for LMIs in finite-dimensional systems, this approach can then be extended to controller and observer synthesis, as outlined below.
\paragraph{Controller Synthesis} For controller synthesis, again consider the LMI approach for the finite-dimensional system:
\[
\dot{x}(t)=Ax(t)+Bu(t),
\]
where $x(t) \in \R^n$ and $u(t) \in \R^m$. For this system, there exists a stabilizing state feedback controller of the form $u(t)=Rx(t)$ if and only if there exists a positive definite matrix $P$ and $Y \in \R^{m \times n}$ such that
\[
(AP+BY)+(AP+BY)^T <0.
\]
If this LMI is feasible, then for $R=YP^{-1}$, the Lyapunov function $V(x)=x^T P^{-1}x$ is positive definite and has time derivative
\begin{align*}
\dot V&=x^T(P^{-1}A+P^{-1}BR + A^T P^{-1}+(BR)^T P^{-1})x\\
&=(P^{-1}x)^T (AP+BRP + PA^T+(BRP)^T)(P^{-1}x)\\
&=y^T (AP+BY + PA^T+(BY)^T)y<0,
\end{align*}
where $y=P^{-1}x$. The extension of this LMI approach to PDEs is to search for a positive definite operator $\mathcal{P}=\mathcal{Z}^*P\mathcal{Z}$ for some $P> 0$ and operator $\mathcal{Y}$, defined by $(\mathcal{Y}z)(z):=Y_1 z(1) + \igzo Y_2(x)z(x)dx$, such that if $u=\mathcal{R}w=\mathcal{Y}\mathcal{P}^{-1}w$,
\begin{align*}
u(t) &= R_1 w(1,t)+\igzo R_2(x)w(x,t)dx\\
&= Y_1 (\mathcal{P}^{-1}w)(1,t)+\igzo Y_2(x)(\mathcal{P}^{-1}w)(x,t)dx,
\end{align*}
then the Lyapunov function $V=\ip{w}{\mathcal{P}^{-1}w}=\ip{\mathcal{Z}(\mathcal{P}^{-1} w)}{P\mathcal{Z}(\mathcal{P}^{-1} w)}$ satisfies $\dot V(w(t))+2\mu V(w(t))=-\ip{\mathcal{Z}\left( \pinv w\right)}{Q\mathcal{Z}\left( \pinv w\right)}$ for some scalar $\mu>0$ and $Q>0$, which implies the closed-loop system is exponentially stable. This is detailed in Section~\ref{sec:synthesis}.

\paragraph{Observer Synthesis}
As mentioned previously, for observer design, we use a Luenberger observer and a separation principle to decouple the error dynamics as defined in Equations~\eqref{eqn:exist:error_form}~-~\eqref{eqn:exist:error_form_BC}. For a finite-dimensional Luenberger observer, where the output is $v(t)=Cx(t)$, the estimator dynamics are defined using the controller gain $F$ and observer gain $L$ as
\[
\dot{\hat{x}}= (A+LC)\hat{x}-L v(t)+Bu(t).
\]
If $u(t)=F\hat x(t)$, then the error dynamics become
\[
\dot e(t)=(A+LC)e(t).
\]
Existence of an observer gain which renders the error dynamics stable is equivalent to the existence of a $P>0$ and $T$ such that
\[
PA+TC +A^TP +C^TT^T<0.
\]
If this LMI is feasible, then for $L=P^{-1}T$, the Lyapunov function $V(e)=e^T P e$ is positive definite and has derivative
\begin{align*}
\dot V(e)&=e^T(PA+PLC + A^T P+C^T L^T P)e\\
&=e^T(PA+TC + A^T P+C^T T^T)e<0.
\end{align*}
For the infinite-dimensional PDE, we have two observer gains which we construct as
\[
L_1(x)=\mathcal{P}^{-1}(T_1(x)+T_3(x)) \quad \text{and}\quad L_2=\mathcal{P}^{-1}(T_2),
\]
 for some gains $T_1$, $T_2$ and $T_3$ and where $\mcl{P}=\mcl{Z}^\star P \mcl{Z}$ for some $P>0$. We then use the Lyapunov function $V(e)=\ip{\mathcal{Z}(e)}{P \mathcal{Z}(e)}$ and search for a $Q>0$ such that $\dot V(e)=-\ip{\mathcal{Z}(e)}{Q\mathcal{Z}(e)}\leq -\delta V(e)<0$, for some $\delta>0$. This is detailed in Section~\ref{sec:obsynth}.

\section{Sum-of-Squares Lyapunov Functions with Semi-Separable Kernels}\label{sec:posop}
In this Section, we define the map $\mathcal Z$ and show how this map is used to construct Lyapunov functions of the form $V(w)=\ip{\mathcal{Z}(w)}{P\mathcal{Z}(w)}$. This approach is based on prior work, as described in~\cite{peetlmi}. Specifically, we define
\[
(\mathcal{Z}w)(x)=\bmat{Z_{d_1}(x)w(x)\\ \int_x^1 Z_{d_2}(x,\xi) w(\xi) d\xi \\ \int_0^x Z_{d_2}(x,\xi) w(\xi)d\xi},
\]
where recall $Z_{d_1}(x)$ and $Z_{d_2}(x,\xi)$ are the vectors of all monomials of degree $d_1$ and $d_2$ or less, starting with $1$.
\begin{theorem}\label{thm:jointpos}
Given $d_1, d_2 \in \mathbb{N}$ and $\epsilon > 0$, $\epsilon \in \mathbb{R}$, let $Z_1(x) = Z_{d_1}(x)$ and $Z_2(x,\xi) = Z_{d_2}(x,\xi)$, with $n=d_1+1$ and $m=\hlf (d_2+2)(d_2+1)$ denoting the length of these vectors, respectively.
Suppose that there exists a matrix $P\in \S^{n+2m}$ such that
\begin{equation}\label{eqn:jointpos:const1}
P=\left[\begin{array}{ccc} P_{11}-\bmat{\epsilon&0_{1,n-1}\\0_{n-1,1}&0_{n-1,n-1}}  & P_{12} & P_{13} \\
P_{12}^T & P_{22} & P_{23} \\
P_{13}^T & P_{23}^T & P_{33}
\end{array} \right]
 \ge 0,\end{equation} where $P_{ij}$ is a partition of $P$ such that $P_{11}\in \S^{n},P_{22}\in \S^{m}$ and $P_{33}\in \S^{m}$. Now let
\begin{align}
&\label{eqn:jointpos:const2}M(x) = Z_{1}(x)^T P_{11}Z_{1}(x),\\
&K_1(x,\xi) = Z_{1}(x)^T P_{12}Z_{2}(x,\xi) + Z_{2}(\xi,x)^T P_{31}Z_1(\xi) \notag \\
&+\int_0^\xi Z_{2}(\eta,x)^T P_{33}Z_{2}(\eta,\xi)d\eta  \notag \\
&+\int_\xi^x Z_{2}(\eta,x)^T P_{32}Z_{2}(\eta,\xi)d \eta \notag \\
&\label{eqn:jointpos:const3} +\int_x^1 Z_{2}(\eta,x)^T P_{22}Z_{2}(\eta,\xi)d\eta,\\
&\label{eqn:jointpos:const4}K_2(x,\xi) = K_1(\xi,x).
\end{align}
Then
\begin{align}
V(w)&=\int_0^1 w(x)M(x)w(x) dx \notag \\
& \qquad  + \int_0^1 \int_0^x w(x) K_1(x,\xi) w(\xi) d \xi dx \notag \\
&\label{eqn:Lyapunov} \qquad \qquad  + \int_0^1 \int_x^1 w(x)K_2(x,\xi) w(\xi) d \xi dx \\
&=\ip{\mathcal{Z}(w)}{P\mathcal{Z}(w)}=\ip{P^\half\mathcal{Z}(w)}{P^\half\mathcal{Z}(w)}\ge \epsilon \norm{w}^2.\notag
\end{align}
\end{theorem}

\begin{proof}
The proof follows directly from the definition of $\mathcal{Z}$ and the Sum-of-Squares representation of $V$.
\end{proof}

The form of the Lyapunov function defined by Theorem~\ref{thm:jointpos} in Equation~\eqref{eqn:Lyapunov} is somewhat atypical for the study of parabolic PDEs. A more commonly used version would be $V(w)=\int_0^1 w(x)M(x)w(x)dx$ or even yet $V(w)=\int_0^1 w(x)M w(x)dx$ for $M>0$. Such forms can be obtained as a special case of Theorem~\ref{thm:jointpos} when $P_{ij}=0$ for $i\neq j \neq 1$. However, as we discuss in Section~\ref{sec:simpler}, neglect of the $K_1$ and $K_2$ terms results in significantly less accurate conditions for stability and control.

For polynomials $M$, $K_1$ and $K_2$, let $\mathcal{X}_{\{M,K_1,K_2\}}$ be defined as in~\eqref{eqn:X}. If $M$, $K_1$ and $K_2$ satisfy the conditions of Theorem~\ref{thm:jointpos}, then $V(w)=\ip{w}{\mathcal{X}_{\{M,K_1,K_2\}}w}\ge \epsilon \norm{w}^2$, which implies the operator $\mathcal{X}_{\{M,K_1,K_2\}}$ is positive definite and furthermore, coercive. Moreover, since $M$, $K_1$ and $K_2$ are polynomials, the operator is bounded, which implies that there exists a $\theta>0$ such that $\epsilon \norm{w}^2 \le V(w)\le \theta \norm{w}^2$. Finally, the constraint~\eqref{eqn:jointpos:const4} in Theorem~\ref{thm:jointpos} implies that the operator $\mcl{X}_{\{M,K_1,K_2\}}$ is self-adjoint.

As discussed in Section~\ref{sec:framework}, Theorem~\ref{thm:jointpos} allows us to use positive matrices to parameterize positive Lyapunov functions of the Form~\eqref{eqn:Lyapunov}. By expanding these forms, the coefficients of the polynomials $M$, $K_1$ and $K_2$ are linear combinations of the elements of $P>0$. Furthermore, if we can express the derivative $\dot V$ in the Form~\eqref{eqn:Lyapunov}, where the coefficients are again linear combinations of the elements of $P$, then we can enforce negativity of the derivative along the solutions $w$ by using $\dot V(w)=-\ip{\mathcal{Z}(w)}{Q\mathcal{Z}(w)}$ to equate these coefficients to those defined by $Q>0$. Constructing the matrices which relate the elements of $P$ and $Q$ can be automated using MATLAB toolboxes for polynomial manipulation such as MULTIPOLY, contained in the package SOSTOOLS~\cite{prajna2001introducing} and further developed in our package DELAYTOOLS~\cite{peetlmi}.

For polynomials $M$, $K_1$ and $K_2$, we represent the constraint $\ip{w}{\mathcal{X}_{\{M,K_1,K_2\}}w}= \ip{\mathcal{Z}(w)}{P\mathcal{Z}(w)}$ for some $P>0$ as $\{M, K_1, K_2\}\in \Xi_{\{d_1,d_2,\epsilon\}}$ where
\begin{align*}
 &\Xi_{\{d_1,d_2,\epsilon\}} :=\{ M,K_1,K_2 \, : \, M,K_1,K_2 \text{ satisfy} \\
 & \qquad \qquad \qquad \qquad \qquad \qquad \qquad  \text{Theorem~\ref{thm:jointpos} for $d_1,d_2,\epsilon$}\}.
\end{align*}
The constraint $\{M, K_1, K_2\}\in \Xi_{\{d_1,d_2,\epsilon\}}$ is an LMI constraint in the coefficients of the polynomials $M$, $K_1$ and $K_2$ and the unknown matrix $P > 0$. In this way, the shorthand $\{M, K_1, K_2\}\in \Xi_{\{d_1,d_2,\epsilon\}}$ allows us to define LMI constraints implicitly.


\section{A Test for Stability}\label{sec:stability}
In this section, we use the results of the previous section to test the existence of a Lyapunov function which establishes stability of the scalar parabolic PDE defined in Equations~\eqref{eqn:prob:PDE_form}~-~\eqref{eqn:prob:PDE_form_BC}.
Recall the autonomous ($u(t)=0$) form of the PDE
\begin{align}
&\label{eqn:stab:PDE_form} w_t(x,t)=a(x)w_{xx}(x,t)+b(x)w_x(x,t)+c(x)w(x,t), \\
&\label{eqn:stab:PDE_form_BC} w(0,t)=0, \qquad w_x(1,t)=0.
\end{align}
The main technical contribution of this section is reformulating the derivative of the Lyapunov function $V$ in~\eqref{eqn:Lyapunov} in the form of Equation~\eqref{eqn:Lyapunov}. This is achieved in the following theorem wherein we obtain functions $\hat M$, $\hat K_1$ and $\hat K_2$ such that
\begin{align*}
\dot V(w)&\le
\int_0^1 w(x)\hat M(x)w(x) dx  \\
& \qquad \qquad + \int_0^1 \int_0^x w(x) \hat K_1(x,\xi) w(\xi) d \xi dx \\
&\qquad \qquad \qquad \qquad + \int_0^1 \int_x^1 w(x) \hat K_2(x,\xi) w(\xi) d \xi dx.
\end{align*}
Note that the inequality in this expression is deliberate, i.e., certain negative semidefinite terms have been left out of $\hat M$, $\hat K_1$ and $\hat K_2$.

Before giving the main theorem, we define the following linear map, $\Omega_s$, which relates functions $M$, $K_1$ and $K_2$ to an upper bound on the time-derivative of the Lyapunov function defined by these functions. Specifically, we say that
\begin{equation}\label{eqn:omega_s}
\{\hat M, \hat K_1, \hat K_2\}:=\Omega_s(M,K_1,K_2),
\end{equation}
if
\begin{align}
\hat{M}(x)=& \pfx \left[\pfx a(x)M(x)-b(x)M(x) \right] \notag \\
& \qquad  + 2\left[\pfx \left[a(x)\left(K_1(x,\xi)-K_2(x,\xi) \right) \right] \right]_{\xi=x}  \notag \\
&\label{eqn:obs:Mhat} \qquad \qquad  +2M(x)c(x)-\frac{\pi^2}{2}\alpha \epsilon,\\
\hat{K}_1(x,\xi)=&  \pfx \left[\pfx \left[ a(x)K_1(x,\xi)\right]-b(x)K_1(x,\xi) \right] \notag \\
& \qquad   +\pfxi \left[\pfxi \left[a(\xi)K_1(x,\xi)\right]-b(\xi)K_1(x,\xi) \right] \notag \\
&\label{eqn:obs:K1hat} \qquad \qquad   +\left(c(x)+c(\xi) \right)K_1(x,\xi),\\
\hat{K}_2(x,\xi)=&\label{eqn:obs:K2hat}\hat{K}_1(\xi,x).
\end{align}
\begin{theorem}\label{thm:analysis}
Suppose that there exist scalars $\epsilon,\delta>0$, $d_1,d_2,\hat{d}_1,\hat{d}_2 \in \N$ and polynomials $M$, $K_1$ and $K_2$ such that
\begin{align*}
&\{M,K_1,K_2\} \in \Xi_{d_1,d_2,\epsilon},\\
&\{-\hat{M}-2\delta M,-\hat{K}_1-2\delta K_1,-\hat{K}_2-2\delta K_2\} \in \Xi_{\hat{d}_1,\hat{d}_2,0},\\
& (b(1)-a_x(1))K_1(1,x)-a(1)(D_1K_{1})(1,x)=0,\\
& (b(1)-a_x(1))M(1)-a(1)M_x(1) \leq 0,\\
& K_2(0,x)=0,
\end{align*} where $\{\hat M, \hat K_1, \hat K_2\}:=\Omega_s(M,K_1,K_2)$.
Then for any initial condition $w(0) \in \mcl{D}_0$, there exists a scalar $\gamma>0$ such that the classical solution $w$ of \eqref{eqn:stab:PDE_form}~-~\eqref{eqn:stab:PDE_form_BC} satisfies
\[\norm{w(t)} \leq \gamma \norm{w(0)}e^{-\delta t}, \quad t > 0,\] where $\mcl{D}_0$ is defined in Equation~\eqref{eqn:D0}.
\end{theorem}
\begin{proof}
Recall the operator $\mcl{X}_{\{M,K_1,K_2\}}$ is as defined in~\eqref{eqn:X}. As discussed in Section~\ref{proset}, for any $w(0) \in \mcl{D}_0$ the autonomous system admits a unique classical solution.
By Theorem~\ref{thm:jointpos}, if $\{M,K_1,K_2\} \in \Xi_{d_1,d_2,\epsilon}$, then
\begin{align*}
&V(w)=\ip{w}{\mcl{X}_{\{M,K_1,K_2\}}w}=\ip{w}{\mcl{P} w},
\end{align*}
satisfies  $\epsilon \norm{w}^2\le V(w)\le \theta \norm{w}^2$ for some $\theta>0$. The calculation of the time derivative $\dot V$ and its reformulation is lengthy. It involves integration by parts, the Wirtinger inequality and the assumption $a(x) \geq \alpha$. For this reason, we have included this proof in the appendix as Lemma~\ref{lem:appendix_1}. Continuing, by Lemma~\ref{lem:appendix_1}, for any $w$ which satisfies Equations~\eqref{eqn:stab:PDE_form}~-~\eqref{eqn:stab:PDE_form_BC},
\begin{align*}
&\dot V(w(t))\le \ip{w(t)}{\mcl{X}_{\{\hat{M},\hat{K}_1,\hat{K}_2\}}w(t)}.
\end{align*}

Now, since $\{-\hat{M}-2\delta M,-\hat{K}_1-2\delta K_1,-\hat{K}_2-2\delta K_2\} \in \Xi_{\hat{d}_1,\hat{d}_2,0}$, we have that $\mcl{X}_{\{\hat{M},\hat{K}_1,\hat{K}_2\}} \leq -2 \delta \mcl{P}$ and thus $-\dot V(w)-2 \delta V(w) \ge 0$. This implies that $\frac{d}{dt} V(w(t))\le -2 \delta V(w(t))$ for all $t\ge 0$. Thus, $V(w(t))\le V(w(0))e^{-2\delta t}$. Concluding, we have that
\[
\norm{w(t)} \leq \gamma \norm{w(0)}e^{-\delta t}, \quad \gamma=\sqrt{\frac{\theta}{\epsilon}}.
\]
\end{proof}
Note that using the arguments in the proof of~\cite[Theorem~$5.1.3$]{curtain1995introduction}, the above result holds for weak/mild solutions where the initial condition need only satisfy $w_0 \in \lt$.

To test the conditions of Theorem~\ref{thm:analysis}, the variables are the coefficients of the polynomials $M$, $K_1$ and $K_2$. The coefficients of $\hat M$, $\hat K_1$ and $\hat K_2$ are then linear combinations of these variables. Finally, the constraints $\in \Xi_{d_1,d_2,\epsilon}$ are LMI constraints, as discussed in Section~\ref{sec:posop}. Constructing the matrices which map these coefficients can be automated using SOSTOOLS or DelayTOOLs. The algorithm used can be adapted from the algorithm presented for output feedback controller in Section~\ref{sec:LMI_conditions}. Application of the conditions of Theorem~\ref{thm:analysis} to several numerical examples can be found in Section~\ref{sec:num_results}.


\section{Inversion and State Transformation}\label{sec:operators}
As discussed in Section~\ref{sec:framework}, for controller synthesis, we will use a state variable transformation $z=\mathcal{P}^{-1}w$ so that $\ip{\mathcal{Z}h}{P\mathcal{Z}(\mathcal{P}^{-1}w)}=\ip{h}{w}$.
Define $\mathcal{P}=\mcl{X}_{\{M,K_1,K_2\}}$, where $\mcl{X}_{\{M,K_1,K_2\}}$ is as defined in~\eqref{eqn:X}. Then $\mathcal{P}$ has the form
\begin{align*}
(\mathcal{P}z)(x)\hspace{-1mm}=\hspace{-1mm}&M(x)z(x) \hspace{-1mm}+ \hspace{-1mm} \int_0^x \hspace{-1mm} K_1(x,\xi) z(\xi) d \xi \hspace{-1mm} + \hspace{-1mm} \int_x^1 \hspace{-1mm} K_2(x,\xi) z(\xi) d \xi,
\end{align*}
where if $\{M, K_1, K_2\}\in \Xi_{\{d_1,d_2,\epsilon\}}$, the operator is coercive with $\ip{w}{\pop w}\ge \epsilon\norm{w}^2$. Operators of this type are a combination of a multiplier operator and two integral operators. Furthermore, since $K_1$ and $K_2$ are polynomials, there exist polynomials $F_i$ and $G_i$ such that $K_1(x,\xi)=F_1(x)^T G_1(\xi)$ and $K_2(x,\xi)=F_2(x)^T G_2(\xi)$. This implies that the two integral operators can be combined into a single integral of the form $\int_0^1K(x, \xi)z(\xi)d\xi$ where $K$ is a kernel of the semiseparable type. That is, there exist functions $F_i$ and $G_i$ such that
\[
K(x,\xi)=\begin{cases} F_1(x)^T G_1(\xi),& x \ge \xi\\
F_2(x)^T G_2(\xi),& \text{otherwise}\end{cases}.
\]
Integral operators with semiseparable kernels are used to represent the input-output map of well-posed Linear Time-Varying (LTV) systems, as explored in~\cite[Section~\Rmnum{1}.$4$, Theorem~$4.1$]{gohberg1984time}. These operators have certain properties which make them well-suited for use in Lyapunov functions. Specifically, they are not trace-class, which means that their eigenvalues may not be summable. Moreover, as discussed in~\cite[Section~\Rmnum{2}.$2$]{gohberg1984time}, since $M(x) \geq \epsilon>0$, $\mathcal{P}^{-1}$ is a bounded linear operator and can be calculated explicitly, as in the following theorem, which is adapted from~\cite[Section~\Rmnum{2}.$3$, Theorem~$3.1$]{gohberg1984time}.
\begin{theorem}\label{thm:invop}
Suppose that $\{M,K_1,K_2\} \in \Xi_{\{d_1,d_2,\epsilon\}}$ for some $d_1,d_2, \epsilon>0$ with $K_1(x,\xi)=F(x)^TG(\xi)$ and $K_2(x,\xi)=G(x)^T F(\xi)$. Let $\pop \in \mcl{L}(\lt)$ be defined as $\pop=\mcl{X}_{\{M,K_1,K_2\}}$, where $\mcl{X}_{\{M,K_1,K_2\}}$ is as defined in~\eqref{eqn:X}. Define
\begin{align*}
B(x)&=\bmat{G(x) \\ F(x)}, \quad C(x)=\bmat{F(x)^T & -G(x)^T},\\
 H&=\left[N_1+N_2U(1) \right]^{-1}N_2 U(1)\\
N_1&=\bmat{I & 0 \\ 0 & 0}, \quad N_2=\bmat{0 & 0 \\ 0 & I},
\end{align*}
and $U(x)=\lim_{n\rightarrow \infty}U_n(x)$, where
\begin{align}
&\label{eqn:inv:fundamental_matrix}
U_{n+1}(x)=I-\igzx B(\xi)M(\xi)^{-1}C(\xi)U_n(\xi)d\xi,
\end{align} and $U_1=I$.
Then, the inverse of the operator $\pop$ is given by
\begin{align*}
\left(\pinv w \right)(x) =& \underbar{M}(x)w(x) + \igzx \underbar{K}_1(x,\xi)w(\xi)d\xi  \\
& \qquad \qquad \qquad \qquad \qquad  +\igxo \underbar{K}_2(x,\xi)w(\xi)d\xi,\\
\underbar{M}(x)=&M(x)^{-1},\\
\underbar{K}_1(x,\xi)=&M(x)^{-1}C(x)U(x)(H-I)U(\xi)^{-1}B(\xi)M(\xi)^{-1},\\
\underbar{K}_2(x,\xi)=&M(x)^{-1}C(x)U(x)HU(\xi)^{-1}B(\xi)M(\xi)^{-1}.
\end{align*} \end{theorem}

Note that since $M(\xi)\ge \epsilon$, $M(\xi)^{-1}$ is bounded and continuous and hence the matrix of rational functions  $B(\xi)M(\xi)^{-1}C(\xi)$ is bounded and continuous. Therefore, it follows from~\cite[Chapter~$3$]{daleckii2002stability} that the uniform limit $U(x)$ exists and is non-singular for $x \in [0,1]$. Since $U(x)$  is non-singular on $[0,1]$, the matrix $H$ is well defined. Therefore, by construction $\underbar{M}, \underbar{K}_1,\underbar{K}_2 \in C^\infty$. Furthermore, note that since $\pop$ satisfies $\epsilon \norm{w}^2 \le \ip{w}{\pop w}\le \theta \norm{w}^2$ for some $\theta >0$, then $1/\theta \norm{w}^2 \le \ip{w}{\pinv w}\le 1/\epsilon \norm{w}^2$.

Theorem~\ref{thm:invop} not only proves existence, but gives a practical method for constructing the state variable transformation $\mathcal{P}^{-1}$ for which $\ip{\mathcal{Z}h}{P\mathcal{Z}(\mathcal{P}^{-1}w)}=\ip{h}{w}$. Specifically, if we truncate the sequence $U_n$ and approximate $M(x)^{-1}$ by a Chebyshev series, then construction of the functions $\underbar M$, $\underbar K_1$ and $\underbar K_2$ is simply a matter of polynomial multiplication and integration, which can be performed in MATLAB or Mathematica. In practice, we have found that $U_n$ converges after only a few iterations. To illustrate, in Figure~\ref{fig:opinverse1} we have applied this approach to a given $\{M,K_1,K_2\} \in \Xi_{1,1,1}$ and plot $\|w-\pop {\pop}_{n+1}^{-1}w\|$ as a function of $n$ for the arbitrarily chosen function $w(x)=x(x-0.4)(x-1)$. Here $\pop_{n+1}^{-1}$ denotes the construction for $\pinv$ defined in Theorem~\ref{thm:invop} with $U(x)$ replaced by $U_{n+1}(x)$. In this case, $n=5$ yields an $L_2$ norm error of $\approx 10^{-5}$. In this example, we approximated $M(\xi)^{-1}$ using the first five terms of its Chebyshev series.

Finally, we emphasize that construction of $\pinv$ is not part of the optimization algorithm, but rather is performed after the algorithm has solved the controller synthesis problem (to be defined in the following section) and returned the polynomial variables $M$, $K_1$, and $K_2$.
\begin{figure}[t]
\centering
\includegraphics[scale=0.25]{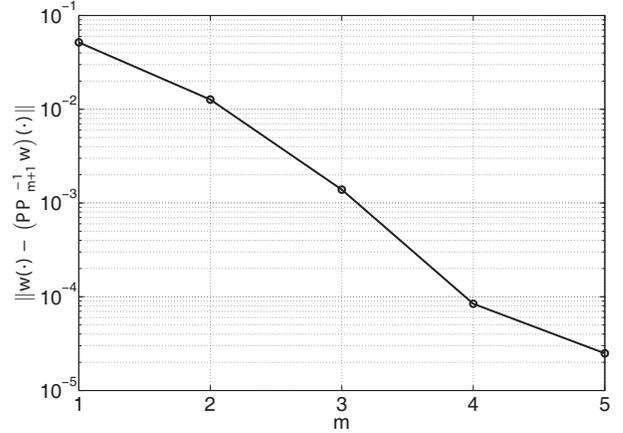}
\caption{$\|w-\mathcal{P}\mathcal{P}_{n+1}^{-1}w\|$ as a function of $n$.}
\label{fig:opinverse1}
\end{figure}


\section{State-Feedback Controller Synthesis}\label{sec:synthesis}
Our approach to controller synthesis is based on the use of a state variable transformation $y=\mathcal{P}^{-1}w$ which, by Theorem~\ref{thm:invop}, is guaranteed to exist for any $\mathcal{P}=\mcl{X}_{\{M,K_1,K_2\}}$ defined by $\{M, K_1, K_2\}\in \Xi_{\{d_1,d_2,\epsilon\}}$. Specifically, we will use the Lyapunov function
$V(w)=\ip{\mathcal{P}^{-1}w}{w}=\ip{y}{\mathcal{P}y}$. Ignoring the input for the moment and using the operator $\mathcal{A}$ defined in Equation~\eqref{eqn:exist:A}, the time-derivative of this function yields the dual stability condition
\[
\dot V(w)=2\ip{\mathcal{P}^{-1}w}{\mathcal{A} w}=2\ip{y}{\mathcal{A} \mathcal{P}y}\le 0,
\]
which we must enforce for all $y \in L_2$. The critical point is that the operator $\mathcal{P}^{-1}$ does not appear explicitly in the stability condition. Rather its existence is only inferred from the constraint on $\mathcal{P}$ that $\{M, K_1, K_2\}\in \Xi_{\{d_1,d_2,\epsilon\}}$. The next step in our approach is to combine this dual stability condition with a variable substitution through the use of a controller of the form
\begin{align*}
u(t) &= Y_1 (\mathcal{P}^{-1}w)(1,t)+\igzo Y_2(x)(\mathcal{P}^{-1}w)(x,t)dx\\
&= R_1 w(1,t)+\igzo R_2(x)w(x,t)dx,
\end{align*}
wherein we have replaced the original controller gains $R_1$ and $R_2$ with the new variables $Y_1$ and $Y_2$. Once $Y_1$ and $Y_2$ are determined by the SOS solver, the actual gains $R_1$ and $R_2$ can be recovered by computing $\mathcal{P}^{-1}$ and applying the formula listed here.

Before giving the main theorem, we recall that the input enters the dynamics as
\begin{align}
&\label{eqn:synth:PDE_form} w_t(x,t)=a(x)w_{xx}(x,t)+b(x)w_x(x,t)+c(x)w(x,t), \\
&\label{eqn:synth:PDE_form_BC} w(0,t)=0, \qquad w_x(1,t)=u(t).
\end{align}
The goal, then, is to define conditions on $P$ (which defines $M$, $K_1$ and $K_2$) as well as on $Y_1$ and the polynomial $Y_2$ such that the closed-loop system is exponentially stable.

To simplify exposition, we now define the following linear map, $\Omega_c$, which relates functions $M$, $K_1$ and $K_2$ to an upper bound on the time-derivative of the Lyapunov function defined by these functions for the controller dynamics. Specifically, we say that
\begin{equation}\label{eqn:omega_c}
\{\hat M, \hat K_1, \hat K_2\}:=\Omega_c(M,K_1,K_2),
\end{equation}
if
\begin{align}
\hat{M}(x)=&\;\left(a_{xx}(x)-b_x(x) \right)M(x)+b(x)M_x(x) \notag \\
&\quad  +a(x)M_{xx}(x)+2c(x)M(x)-\frac{\pi^2}{2}\alpha \epsilon \notag \\
&\label{eqn:synth:Mhat}  \quad \quad  +a(x) \left[2\pfx \left[K_1(x,\xi)-K_2(x,\xi) \right] \right]_{\xi=x}, \\
\hat{K}_1(x,\xi)=&\;a(x)(D_1^2 K_{1})(x,\xi)+b(x)(D_1 K_{1})(x,\xi) \notag \\
& \quad \quad +a(\xi)(D_2^2 K_{1})(x,\xi)+b(\xi)(D_2 K_{1})(x,\xi) \notag \\
&\label{eqn:synth:K1hat}\quad \quad \quad \quad +\left(c(x)+c(\xi) \right)K_1(x,\xi),\\
\hat{K}_2(x,\xi)=&\label{eqn:synth:K2hat}\;\hat{K}_1(\xi,x).
\end{align}
\begin{theorem}\label{thm:synthesis}
Suppose that there exist scalars $\epsilon,\mu>0$, $d_1,d_2,\hat{d}_1,\hat{d}_2 \in \N$ and polynomials $M$, $K_1$ and $K_2$ such that
\[
\{M,K_1,K_2\} \in \Xi_{d_1,d_2,\epsilon} \quad \text{and}\quad  K_2(0,x)=0.
\]
Further suppose
\[
\{-\hat{M}-2\mu M,-\hat{K}_1-2\mu K_1,-\hat{K}_2-2\mu K_2\} \in \Xi_{\hat{d}_1,\hat{d}_2,0},
\]
where $\{\hat M, \hat K_1, \hat K_2\}=\Omega_c(M,K_1,K_2)$. Let
\begin{equation}\label{eqn:synth:Y}
Y_1<\frac{M_x(1)}{2}+\frac{a_x(1)-b(1)}{2a(1)}M(1),\quad Y_2(x)=(D_1 K_{1})(1,x).
\end{equation}
If the control input $u(t)$ is defined as
\begin{align}
u(t) =& Y_1 (\mathcal{P}^{-1}w)(1,t)+\igzo Y_2(x)(\mathcal{P}^{-1}w)(x,t)dx \notag\\
 =&\label{eqn:synth:control} R_1 w(1,t)+\igzo R_2(x)w(x,t)dx,
 \end{align}
where $\mathcal{P}^{-1}$ is as defined for $\mcl{P}=\mcl{X}_{\{M,K_1,K_2\}}$ in Theorem~\ref{thm:invop} and $\mcl{X}_{\{M,K_1,K_2\}}$ is as defined in~\eqref{eqn:X}, then there exists a scalar $\gamma>0$ such that for any initial condition $w(0) \in \mcl{D}$ (where $\mcl{D}$ is as in Equation~\eqref{eqn:D}) the solution $w$ of \eqref{eqn:synth:PDE_form}~-~\eqref{eqn:synth:PDE_form_BC} exists, belongs to $C^{1,2}((0,\infty),[0,1])$ and satisfies
\[
\norm{w(t)} \leq \gamma \norm{w(0)} e^{-\mu t}, \quad t > 0.
\]
\end{theorem}
\begin{proof}
We start the proof by observing that since $\{M,K_1,K_2\}\in \Xi_{d_1,d_2,\epsilon}$, as per Theorem~\ref{thm:jointpos}, these polynomials define a positive operator $\pop=\mcl{X}_{\{M,K_1,K_2\}}$ such that $\epsilon \norm{w}^2 \le \ip{w}{\pop w}\le \theta \norm{w}^2$ for some $\theta>0$. Furthermore, by Theorem~\ref{thm:invop}, there exist bounded and continuously differentiable functions $\underbar M$, $\underbar K_1$ and $\underbar K_2$ such that $\pinv=\mcl{X}_{\{\underbar M,\underbar K_1, \underbar K_2\}}$ satisfying $1/\theta \norm{w}^2 \le \ip{w}{\pinv w} \le 1/\epsilon \norm{w}^2$. We now propose the Lyapunov function
\begin{align*}
V(w)&=\ip{\mathcal{P}^{-1} w}{w}=\ip{\mathcal{P}^{-1}w}{\mathcal{P}\mathcal{P}^{-1}w}\\
&=\int_0^1 (\mathcal{P}^{-1}w)(x)M(x)(\mathcal{P}^{-1}w)(x) dx \\
& \qquad + \int_0^1 \int_0^x (\mathcal{P}^{-1}w)(x) K_1(x,\xi) (\mathcal{P}^{-1}w)(\xi) d \xi dx \notag \\
&\qquad \qquad   + \int_0^1 \int_x^1 (\mathcal{P}^{-1}w)(x)K_2(x,\xi) (\mathcal{P}^{-1}w)(\xi) d \xi dx.
\end{align*}
Let $y=\pinv w$. Note that if $w \in H^2(0,1)$, then $y=\pinv w \in H^2(0,1)$.
Now, since $1/\theta \norm{w}^2 \le \ip{w}{\pinv w} \le 1/\epsilon \norm{w}^2$, we have that the Lyapunov function is upper and lower bounded. Now suppose that
\begin{align*}
u(t) &= Y_1 (\mathcal{P}^{-1}w)(1,t)+\igzo Y_2(x)(\mathcal{P}^{-1}w)(x,t)dx\\
&= R_1 w(1,t)+\igzo R_2(x)w(x,t)dx.
\end{align*}
Since $\underbar{M},\underbar{K}_1,\underbar{K}_2 \in C^\infty(0,1)$ and $Y_2$ is polynomial, we have that $R_2 \in C^\infty(0,1)$. Therefore, as discussed in Section~\ref{proset}, the closed loop System~\eqref{eqn:synth:PDE_form}~-~\eqref{eqn:synth:PDE_form_BC} admits a solution $w \in H^2(0,1)$ which implies $y=\pinv w \in H^2(0,1)$.
Again, the calculation of the time derivative $\dot V$ and its reformulation is lengthy. It involves integration by parts, the Wirtinger inequality and the assumption $a(x) \geq \alpha$. This proof is in the appendix as Lemma~\ref{lem:control} which establishes that for any $w$ which satisfies Equations~\eqref{eqn:synth:PDE_form}~-~\eqref{eqn:synth:PDE_form_BC},
\begin{align*}
\dot V(w(t))\le& \ip{y(t)}{\mcl{X}_{\{\hat{M},\hat{K}_1,\hat{K}_2\}}y(t)} \\
&\qquad  +y(1,t)Ny(1,t) + 2y(1,t)a(1)M(1)y_x(1,t),
\end{align*}
where $N=a(1)M_x(1)+(b(1)-a_x(1))M(1)$. Now, since $\{-\hat{M}-2\mu M,-\hat{K}_1-2\mu K_1,-\hat{K}_2-2\mu K_2\} \in \Xi_{d_1,d_2,0}$, we have that $\mcl{X}_{\{\hat{M},\hat{K}_1,\hat{K}_2\}} \leq -2 \mu \pop$ and hence $\ip{y(t)}{\mcl{X}_{\{\hat{M},\hat{K}_1,\hat{K}_2\}}y(t)}\le -2 \mu \ip{y(t)}{\pop y(t)}=-2 \mu \ip{w(t)}{\pinv w(t)}$. Applying this to the inequality, we get
\begin{align*}
\dot V(w(t))\le& -2 \mu \ip{w(t)}{\pinv w(t)} \\
&\qquad+y(1,t)Ny(1,t) + 2y(1,t)a(1)M(1)y_x(1,t).
\end{align*}
A sufficient condition for stability, then, is that $2y(1)a(1)M(1)y_x(1)\le -y(1)N y(1)$. Unfortunately, however, our control input enters via $w_x(1)$ and not $y_x(1)$. To see the relationship between $w_x(1)$ and $y_x(1)$, we expand the former and then solve for the latter as follows
\begin{align}
w_x(1,t)=&M_x(1)y(1,t)+M(1)y_x(1,t) \notag \\
&\label{eqn:synth:inter2} \qquad +\igzo (D_1 K_{1})(1,x)y(x,t)dx,
\end{align}
where solving for $M(1)y_x(1)$ yields
\begin{align}
M(1)y_x(1,t)=&w_x(1,t)-M_x(1)y(1,t) \notag \\
&\label{eqn:synth:inter2}\qquad-\igzo (D_1 K_{1})(1,x)y(x,t)dx.
\end{align}
This implies that the Lyapunov function satisfies
\begin{align*}
\dot V(w(t))\le& -2\mu V(w(t)) + y(1,t)Ny(1,t) \\
&\quad+ 2y(1,t)a(1)\left(w_x(1,t)-M_x(1)y(1,t)\right)\\
&\quad \quad -2y(1,t)a(1)\igzo (D_1 K_{1})(1,x)y(x,t)dx.
\end{align*}
Now, examining the proposed controller, we obtain
\begin{align*}
w_x(1,t)=u(t)&=R_1 w(1,t)+\int_0^1 R_2(x) w(x,t)dx\\
&=R_1 (\pop y)(1,t)+\int_0^1 R_2(x) (\pop y)(x,t)dx\\
&=Y_1 y(1,t)+\int_0^1 Y_2(x) y(x,t)dx,
\end{align*}
which is expressed in the new optimization variables $Y_1$ and $Y_2$.
Now, plugging $w_x(1)=Y_1 y(1)+\int_0^1 Y_2(x) y(x)dx$ into the time-derivative of the Lyapunov function, we get
\begin{align}
\dot V(w(t))\le& -2\mu V(w(t)) \notag \\
&+y(1,t)^2(N+2 a(1) Y_1-2a(1)M_x(1)) \notag \\
&  + 2y(1,t)a(1)\igzo Y_2(x)y(x,t)dx \notag \\
&\label{eqn:LMI:synth}  - 2y(1,t)a(1)\igzo (D_1 K_{1})(1,x)y(x,t)dx.
\end{align}
By inspection, we see that the stability conditions are now
\[
N+2 a(1) Y_1-2a(1)M_x(1)< 0
\]
and $Y_2(x)=(D_1 K_{1})(1,x)$.
This then implies that $\dot{V}(w(t)) \leq  -2\mu V(w(t))$ for all $t \ge 0$ and hence $V(w(t))\le V(w(0))e^{-2\mu t}$. Since $\displaystyle \norm{w}^2 \le \theta V(w)$, we have
\begin{align*}
\norm{w(t)} \le \sqrt{\theta V(w(t))} \leq \sqrt{\theta V(w(0))} &e^{-\mu t}\\
&\le \sqrt{\theta/\epsilon} \norm{w(0)} e^{-\mu t}.
\end{align*}
\end{proof}
At this point, it is significant to note that given values for the variables $Y_1$, $Y_2$, $M$, $K_1$ and $K_2$, the controller gains $R_1$ and $R_2$ can be found by calculating $\underbar{M}$, $\underbar{K}_1$ and $\underbar{K}_2$ via Theorem ~\ref{thm:invop} and using the formula
\begin{align*}
&R_1 w(1)+\int_0^1 R_2(x) w(x)dx \\
&=Y_1 y(1)+\int_0^1 Y_2(x) y(x)dx\\
&=Y_1 (\pinv w)(1)+\int_0^1 Y_2(x) (\pinv w)(x)dx\\
&=Y_1 \underbar{M}(1) w(1)+\int_0^1 Y_1 \underbar K_1(1,x)w(x)dx \\
& \hspace{-1mm}+\hspace{-1mm}\int_0^1Y_2(x)\left( \int_0^x\underbar K_1(x,\xi)w(\xi) d\xi dx \hspace{-1mm} + \hspace{-1mm} \int_x^1 \underbar K_2(x,\xi)w(\xi) d\xi\right) dx\\
&=Y_1 \underbar{M}(1) w(1)+\int_0^1 Y_1 \underbar K_1(1,x)w(x)dx \\
&+\hspace{-1mm} \int_0^1\left( \int_x^1 Y_2(\xi)\underbar K_1(\xi,x) d\xi \hspace{-1mm} + \hspace{-1mm} \int_0^x Y_2(\xi)\underbar K_2(\xi,x)d\xi \right)w(x)  dx,
\end{align*}
so that
\begin{align}
R_1 =&\label{eqn:synth:R1} Y_1 \underbar{M}(1),\\
R_2(x) \hspace{-1mm} =&\label{eqn:synth:R2}Y_1 \underbar K_1(1,x) \hspace{-1mm} + \hspace{-1mm} \int_x^1  \hspace{-1mm} Y_2(\xi)\underbar K_1(\xi,x)d\xi \hspace{-1mm} + \hspace{-1mm} \int_0^x  \hspace{-1mm} Y_2(\xi)\underbar K_2(\xi,x) d\xi,
\end{align}
where we have used the identity
\[
\int_0^1 \hspace{-1mm} \int_x^1 \hspace{-1mm} f(x,\xi)d\xi dx=\int_0^1 \hspace{-1mm} \int_0^\xi \hspace{-1mm} f(x,\xi)dx d\xi=\int_0^1 \hspace{-1mm} \int_0^x \hspace{-1mm} f(\xi,x)d\xi dx,
\] and the fact that $\underbar K_1(x,\xi)=\underbar K_2(\xi,x)$.

\section{Observer Synthesis}\label{sec:obsynth}
In Section~\ref{sec:synthesis}, we described LMI conditions under which one can obtain controller gains $R_1$ and $R_2(x)$ such that input $u(t)=R_1w(1,t)+\igzo R_2(x)w(x,t)dx$ ensures exponentially stability of the closed-loop system. However, this form of controller requires measurement of the state $w(x,t)$ at every point $x \in [0,1]$ at all times. Implementation of such a controller is problematic as such a distributed measurement is unlikely to be readily available. A more common scenario is one in which we may only measure the value of $w(x,t)$ at discrete points in the domain. In particular, we assume that only a single measurement is available at the boundary of the domain, and in particular, at $v(t)=w(1,t)$. As discussed in Section~\ref{proset}, in this scenario, we seek to find an estimator/observer which will yield a real-time estimate of the state of the system at all points and which, if used in closed-loop, will ensure exponential stability of the closed-loop. Specifically, our observer is a dynamic system with input $v(t)=w(1,t)$ and output $\hat w(x,t)$, where $\hat w(x,t)$ is the estimate of the state at time $t$. We adopt the Luenberger observer framework discussed previously, which implies that the dynamics of the observer are given by
\begin{align}
&\wh_t(x,t)=a(x)\wh_{xx}(x,t)+b(x)\wh_x(x,t) \notag \\
&\label{eqn:obs:observer_form}\qquad \qquad \qquad+c(x)\wh(x,t)+L_1(x)\left(\hat{v}(t)-v(t) \right), \\
&\label{eqn:obs:observer_form_BC} \wh(0,t)=0, \qquad \wh_x(1,t)=u(t)+L_2\left(\hat{v}(t)-v(t) \right),
\end{align}
where $\hat{v}(t)=\wh(1,t)$ is the predicted output and the scalar $L_2$ and function $L_1(x)$ are gains which map error in this predicted output to the dynamics of the observer state. In the following theorem, we seek conditions on $L_1$ and $L_2$ which ensure that if $R_1$ and $R_2$ are as defined in Theorem~\ref{thm:synthesis} and the controller is defined as
\begin{equation}
u(t)=R_1\wh(1,t)+\igzo R_2(x)\wh(x,t)dx,\label{eqn:obs:control_form}
\end{equation}
then Equations~\eqref{eqn:obs:observer_form}~-~\eqref{eqn:obs:observer_form_BC} coupled with Equations~\eqref{eqn:synth:PDE_form}~-~\eqref{eqn:synth:PDE_form_BC} and Equation~\eqref{eqn:obs:control_form} define an exponentially stable system.

Our approach is based on the separation principle~\cite[Chapter~$5$]{curtain1995introduction},~\cite[Chapter~$5$]{krstic2008boundary}. Specifically, we consider the error dynamics of the PDE coupled with the observer dynamics in Equations~\eqref{eqn:obs:observer_form}~-~\eqref{eqn:obs:observer_form_BC}. That is, if we define the error as $e=\wh-w$, then this quantity satisfies
\begin{align}
& e_t(x,t)=a(x)e_{xx}(x,t)+b(x)e_x(x,t)+c(x)e(x,t) \notag \\
&\label{eqn:obs:error_form} \qquad \qquad \qquad+z_1(x,t),\\
&\label{eqn:obs:error_form_BC} e(0,t)=0, \qquad e_x(1,t)=z_2(t),
\end{align}
where the feedback signals $z_1$ and $z_2$ are defined as
\begin{equation}\label{eqn:obs:zvars}
z_1(x,t):=L_1(x)e(1,t)\quad \text{and} \quad z_2(t):=L_2e(1,t).
\end{equation}
The key point is that the error dynamics do not depend on the choice of controller gains $R_1$ and $R_2$. In the following theorem, this will allow us to choose observer gains $L_1$ and $L_2$ which stabilize the error dynamics. Then, in Theorem~\ref{thm:LMI} we will show that if the controller gains are chosen as per Theorem~\ref{thm:synthesis} and the observer gains are chosen as per Theorem~\ref{thm:observer}, then the coupled dynamics are stable in both the state and state estimate. Unlike for controller synthesis, the conditions for stabilization of the error dynamics are based on the use of a simple Lyapunov function $V(e)=\ip{e}{\pop e}$ where the operator $\pop=\mcl{X}_{\{M,K_1,K_2\}}$ is defined by some $\{M,K_1,K_2\} \in \Xi_{d_1,d_2,\epsilon}$.

The following theorem is motivated by the LMI approach as defined in Section~\ref{sec:framework}, wherein as before the variables $M$, $K_1$ and $K_2$ are defined by a positive definite matrix $P$ and the observer variables are scalar $T_2$ and polynomials $T_1$ and $T_3$ (defined by their vector of coefficients). Referring to the LMI motivation, these observer variables are similar to the matrix $T$ and the observer gains are then recovered as $L_2=\pinv T_2$ and $L_1=\pinv (T_1+T_3)$, which is similar to the LMI observer gain matrix $L=P^{-1}T$.
\begin{theorem}\label{thm:observer}
Suppose there exist scalars $\epsilon,\delta>0$, $d_1,d_2,\hat d_1, \hat d_2 \in \N$ and polynomials $M$, $K_1$ and $K_2$ such that
\[
\{M,K_1,K_2\} \in \Xi_{d_1,d_2,\epsilon},\quad \text{and}\quad  K_2(0,x)=0.
\]
Further suppose
\[
\{-\hat{M}-2\delta M,-\hat{K}_1-2\delta K_1,-\hat{K}_2-2\delta K_2\} \in \Xi_{\hat{d}_1,\hat{d}_2,0},
\]
where $\{\hat M, \hat K_1, \hat K_2\}:=\Omega_s(M,K_1,K_2)$.
Let $\underbar M$, $\underbar K_1$ and $\underbar K_2$ define $\mcl{X}_{\{M,K_1,K_2\}}^{-1}=\mcl{X}_{\{\underbar M, \underbar K_1, \underbar K_2\}}$ as in Theorem~\ref{thm:invop} and
\begin{align}
L_2 \coloneqq& \label{eqn:obs:gain:1} (a(1)M(1))^{-1}T_2,\\
L_1(x)\coloneqq&  \underbar M(x)(T_1(x)+T_3(x)) \notag \\
&\qquad  +\int_0^x\underbar K_1(x,\xi)(T_1(\xi)+T_3(\xi))d \xi \notag \\
&\label{eqn:obs:gain:2}  \qquad \qquad  +\int_x^1\underbar K_2(x,\xi)(T_1(\xi)+T_3(\xi))d \xi,
\end{align}
where
\begin{align}
T_1(x)&\label{eqn:obs:T1}=-0.5( (b(1)-a_x(1))K_1(1,x) -a(1)(D_1 K_{1})(1,x)),\\
T_2&\label{eqn:obs:T2}< - 0.5((b(1)-a_x(1))M(1)-a(1)M_x(1))\\
T_3(x)&\label{eqn:obs:T3}=-L_2a(1)K_1(1,x),
\end{align} and $\mcl{X}_{\{M,K_1,K_2\}}$ and $\mcl{X}_{\{\underbar M, \underbar K_1, \underbar K_2\}}$ are as defined in~\eqref{eqn:X}.
Then for any $e$ which satisfies~\eqref{eqn:obs:error_form}~-~\eqref{eqn:obs:error_form_BC} with initial condition $e(0) \in \mcl{D}_e$ (See Equation~\eqref{eqn:D_e}), there exists a scalar $\gamma>0$
such that
\[
\norm{e(t)} \leq \gamma \norm{e(0)}e^{-\delta t}, \quad t>0.
\]
\end{theorem}
\begin{proof}
We start by observing that since $\{M,K_1,K_2\}\in \Xi_{d_1,d_2,\epsilon}$, as per Theorem~\ref{thm:jointpos}, these polynomials define a positive operator $\pop=\mcl{X}_{\{M,K_1,K_2\}}$ such that $\epsilon \norm{w}^2 \le \ip{w}{\pop w}\le \theta \norm{w}^2$ for some $\theta>0$. Furthermore, by Theorem~\ref{thm:invop}, there exist bounded and continuously differentiable functions $\underbar M$, $\underbar K_1$ and $\underbar K_2$ which define the positive operator $\pinv=\mcl{X}_{\{\underbar M, \underbar K_1, \underbar K_2\}}$. Therefore, since $L_1 \in C^\infty(0,1)$, we have that the closed-loop error dynamics~\eqref{eqn:obs:error_form}~-~\eqref{eqn:obs:error_form_BC} admit a local in time solution $e$ for any $e_0 \in \mcl{D}_e$.

We now propose the Lyapunov function
\begin{align*}
V(e)=&\ip{e}{\mathcal{P} e}\\
=&\int_0^1 e(x)M(x)e(x) dx  \\
&\qquad + \int_0^1 \int_0^x e(x) K_1(x,\xi) e(\xi) d \xi dx \\
&\qquad \qquad+ \int_0^1 \int_x^1 e(x)K_2(x,\xi)e(\xi) d \xi dx.
\end{align*}
The derivative of this Lyapunov function is identical to the one in Theorem~\ref{thm:analysis} except for the presence of the terms \textbf{$z_1$} and \textbf{$z_2$} defined in~\eqref{eqn:obs:zvars}. Specifically, we have
\begin{align*}
\dot V(e)\le& \ip{e}{\mcl{X}_{\{\hat M,\hat K_1,\hat K_2\}}e}  \\
&\qquad +2 \ip{ \pop e}{z_1} +2 z_2 a(1) (\pop e)(1) \\
&\qquad \qquad+2\ip{e}{\mathcal{M}_{R_1}e(1)}  +e(1) R_2 e(1),
\end{align*}
where $R_1(x)= (b(1)-a_x(1))K_1(1,x) -a(1)(D_1 K_{1})(1,x)$ and $R_2=(b(1)-a_x(1))M(1)-a(1)M_x(1)$.
In the proof of Theorem~\ref{thm:analysis}, we had $z_1=0$ and $z_2=0$ and hence the stability condition was that $R_1=R_2=0$ and that $\mcl{X}_{\{\hat M,\hat K_1,\hat K_2\}}\leq - 2\delta \mathcal{P}$. For the observer, we similarly require $\mcl{X}_{\{\hat M,\hat K_1,\hat K_2\}}\leq - 2 \delta \mathcal{P}$. However, we now have the observer gains $z_1(x)=L_1(x)e(1)$ and $e_x(1)=z_2=L_2 e(1)$ which the algorithm can choose in order to cancel out $R_1$ and $R_2$. Unfortunately, however, these gains depend on $M$ and $K_1$ and the gains are currently bilinear with the operator variable $\pop$ (and the functions $M$ $K_1$, and $K_2$ which define it). Hence we would like to perform a variable substitution. This is complicated, however, by the fact that there are two observer gains - one at the boundary and one directly injected into the dynamics. Let us first examine the second gain $z_2=L_2 e(1)$ which appears in the term
\begin{align*}
&z_2 a(1) (\pop e)(1)=e(1)L_2a(1)(\pop e)(1)\\
&\qquad =e(1)\underbrace{L_2a(1)M(1)}_{T_2}e(1) \\
&\qquad \qquad \qquad +\int^1_0 e(1)L_2a(1)K_1(1,x)e(x)dx\\
&\qquad =e(1)T_2e(1) +\int^1_0 e(1)L_2a(1)K_1(1,x)e(x)dx,
\end{align*}
where we have made the variable substitution $T_2=L_2a(1)M(1)$ which implies $T_2$ is a scalar variable. The variable $L_2$ is thereby partially eliminated from the search. However, since $a(x)>0$ and $M(x)>0$, given $T_2$, the gain $L_2$ can later be recovered as $L_2=(a(1)M(1))^{-1}T_2$. Of course, this variable substitution has not \textit{completely} eliminated the original variable $L_2$. To completely eliminate $L_2$ will require assistance from the second gain $L_1$. To see how this is done, we examine the second term in which $\mathbf{z_1}$ appears
\begin{align*}
\ip{ \pop e}{z_1}& \hspace{-1mm} = \hspace{-1mm} \ip{  e}{\underbrace{\pop \mathcal{M}_{L_1}}_{\mathcal{M}_{T_1}+\mathcal{M}_{T_3}}e(1)}\hspace{-1mm} = \hspace{-1mm} \ip{e}{\mathcal{M}_{T_1}e(1)} \hspace{-1mm} + \hspace{-1mm} \ip{e}{\mathcal{M}_{T_3}e(1)}.
\end{align*}
Here we have defined a new variable $T_1(x)$ which is defined by $T_1(x):=M(x)L_1(x)+\int_0^xK_1(x,\xi)L_1(\xi)d\xi +\int_x^1K_2(x,\xi)L_1(\xi)d\xi - T_3(x)$ for which $\mathcal{M}_{T_1}c=\pop \mathcal{M}_{L_1}c-\mathcal{M}_{T_3}c$ for any $c \in \R$ where $T_3$ will be defined shortly. Furthermore, for any $T_3$, the map $L_1 \mapsto T_1$ is invertible with
\begin{align*}
L_1(x):=&\underbar M(x)(T_1(x)+T_3(x)) \\
&\qquad +\int_0^x \underbar{K}_1(x,\xi)(T_1(\xi)+T_3(\xi))d\xi \\
&\qquad \qquad +\int_x^1\underbar{K}_2(x,\xi)(T_1(\xi)+T_3(\xi))d\xi
\end{align*}
if $\pinv=\mcl{X}_{\{\underbar M,\underbar{K}_1,\underbar{K}_2\}}$. In this way, we eliminate the variable $L_1$ and replace it with $T_1$ and $T_3$.
The next step, then, is to choose $T_3$ so as to cancel the remaining term which contains $L_2$. This is done using $\ip{e}{\mathcal{M}_{T_3}e(1)}$, which we expand to get
\begin{align*}
\ip{e}{\mathcal{M}_{T_3}e(1)}=\int_0^1 e(x)T_3(x)e(1)dx,
\end{align*}
which we would like to use to eliminate $\int^1_0 e(1)L_2a(1)K_1(1,x)e(x)dx$. Clearly, then, the appropriate choice for $T_3$ is
\[
T_3(x)=-L_2a(1)K_1(1,x).
\]
Note that the dependence of $T_3$ on $L_2$ is admissible because $T_3$ is not a free variable and $L_2$ is computed directly from $T_2$. This means that once feasible values for $T_1$ and $T_2$ have been found, we then calculate $L_2$ from $T_2$, then use $L_2$ to calculate $T_3$ and then use $T_1$ and $T_3$ to calculate the gain $L_1$.

Concluding the proof, the time-derivative of the Lyapunov function becomes
\begin{align*}
\dot V(e)&\le \ip{e}{\mcl{X}_{\{\hat M,\hat K_1,\hat K_2\}}e} \\
&+2\ip{e}{\left(\mathcal{M}_{T_1}+\mathcal{M}_{R_1}\right)e(1)}+ e(1)\left(2 T_2+R_2\right)e(1).
\end{align*}
Therefore, if $T_1=-R_1$, $2 T_2+R_2< 0$ and $\{-\hat{M}-2\delta M,-\hat{K}_1-2\delta K_1,-\hat{K}_2-2\delta K_2\} \in \Xi_{\hat{d}_1,\hat{d}_2,0}$, we have that
\[
\dot V(e)\le -2\delta V(e),
\]
which, in a similar manner as Theorem~\ref{thm:analysis} establishes exponential stability of the error dynamics with decay rate $\delta$.
\end{proof}

\section{An LMI condition for Output-Feedback Stabilization}\label{sec:LMI_conditions}
In this section we briefly summarize the results of the paper by giving an LMI formulation of the output-feedback controller synthesis problem.

\begin{theorem}\label{thm:LMI}
Given $d_1,d_2,\hat d_1,\hat d_2 \in \N$ and $\epsilon, \delta ,\mu>0$, suppose that there exist polynomials $M$, $N$, $K_1$, $K_2$, $S_1$ and $S_2$ such that
\begin{align}
\{M,K_1,K_2\} &\label{eqn:LMI:cond1}\in \Xi_{d_1,d_2,\epsilon},\\
\{-\hat M - 2\mu M, -\hat K_1 - 2\mu K_1,-\hat K_2 - 2\mu K_2\}&\label{eqn:LMI:cond2}\in \Xi_{\hat d_1,\hat d_2,0},\\
K_2(0,x)&\label{eqn:LMI:cond3}=0,\\
\{N,S_1,S_2\} &\label{eqn:LMI:cond4}\in \Xi_{d_1,d_2,\epsilon},\\
\{-\hat N - 2\delta N, -\hat S_1 - 2\delta S_1,-\hat S_2 - 2\delta S_2\}&\label{eqn:LMI:cond5}\in \Xi_{\hat d_1,\hat d_2,0},\\
S_2(0,x)&\label{eqn:LMI:cond6}=0,
\end{align} where $\{\hat M,\hat K_1,\hat K_2\} =\Omega_c(M,K_1,K_2)$, $\{\hat N,\hat S_1,\hat S_2\} =\Omega_s(N,S_1,S_2)$ and $2\hat d_1$ and $2\hat d_2+1$ are the degrees of $M$, $N$ and $K_1$, $S_1$, respectively.

Then, there exist gains $R_1$, $R_2(x)$, $L_1(x)$ and $L_2$ such that if
\begin{equation}\label{eqn:LMI:control}
u(t)=R_1 \wh(1,t)+\igzo R_2(x)\wh(x,t)dx,
\end{equation} and $w$ satisfies Equations~\eqref{eqn:synth:PDE_form}~-~\eqref{eqn:synth:PDE_form_BC} and $\wh$ satisfies Equations~\eqref{eqn:obs:observer_form}~-~\eqref{eqn:obs:observer_form_BC} with a zero initial condition then
\[
\norm{w(t)} \leq \gamma \norm{w(0)} e^{-\kappa t},
\] for some $\gamma>0$ and any $\kappa$ satisfying $0<\kappa<\min\{\mu,\delta\}$.
\end{theorem}
 \begin{proof}
If the conditions in~\eqref{eqn:LMI:cond1}~-~\eqref{eqn:LMI:cond3} are satisfied, then the polynomials $M$, $K_1$ and $K_2$ satisfy the constraints of Theorem~\ref{thm:synthesis}. Therefore, we may construct $R_1$ and $R_2(x)$ using~\eqref{eqn:synth:R1}~-~\eqref{eqn:synth:R2}. Similarly, if $N$, $S_1$ and $S_2$ satisfy~\eqref{eqn:LMI:cond4}~-~\eqref{eqn:LMI:cond6}, then the conditions of Theorem~\ref{thm:observer} are satisfied with $M=N$, $K_1=S_1$ and $K_2=S_2$. Thus, we can construct observer gains $L_1(x)$ and $L_2$ using~\eqref{eqn:obs:gain:1}~-~\eqref{eqn:obs:gain:2}.
 Now, let $\pop_c=\mcl{X}_{\{M,K_1,K_2\}}$, $\hat \pop_c=\mcl{X}_{\{\hat M,\hat K_1, \hat K_2\}}$, $\pop_o=\mcl{X}_{\{N,S_1,S_2\}}$ and $\hat \pop_o=\mcl{X}_{\{\hat N,\hat S_1,\hat S_2\}}$. Therefore, the theorem conditions imply that $\hat \pop_c \leq -2 \mu \pop_c$ and $\hat \pop_o \leq -2 \delta \pop_o$.

Using the proof of Theorem~\ref{thm:observer}, there exists a scalar $\beta_1>0$ such that
\begin{equation}\label{eqn:LMI:Vdot_obs}
\dot{V}_o(e) \leq -2 \delta V_o(e)-\beta_1 e(1)^2,
\end{equation} where $V_o(e)=\ip{e}{\pop_o e}$. Similarly, for the observer dynamics in~\eqref{eqn:obs:observer_form}~-~\eqref{eqn:obs:observer_form_BC} with the input~\eqref{eqn:LMI:control}, using the proof of Theorem~\ref{thm:synthesis} one can prove that there exists a scalar $\beta_2>0$ such that
 \begin{align}
\dot{V}_c(\wh)\leq& -2 \mu V_c(\wh)+2\ip{\hat y}{L_1 e(1)}+\hat y(1)(2a(1)L_2)e(1) \notag \\
&\label{eqn:LMI:Vdot_synth} -\beta_2 \hat y(1)^2,
\end{align} where $\hat y=\pinv_c \wh$ and $ V_c(\wh)=\ip{\wh}{\pinv_c \wh}=\ip{\hat y}{\pop_c \hat y}$.
 From~\eqref{eqn:LMI:Vdot_obs}~-~\eqref{eqn:LMI:Vdot_synth} we infer that for any $r>0$ we have
\begin{align}
&\label{eqn:LMI:inter}r\dot{V}_o(e)+\dot{V}_c(\wh)\leq -2r \delta V_o(e)+\left\langle \bmat{\hat y\\ \hat y(1) \\ e(1)},\mcl{U} \bmat{\hat y\\ \hat y(1) \\ e(1)} \right\rangle,
\end{align} where
\[\mcl{U}=\bmat{-2\mu \pop_c & 0 & L_1 \\ \star & -\beta_2 & a(1)L_2 \\ \star & \star & -r\beta_1},\] and the inner product is defined on $\lt \times \R \times \R$. Now, for any $0<\kappa<\min\{\delta,\mu\}$, if we choose $r>0$ sufficiently large, it follows that $\mcl{U}\leq \diag(-2\kappa \pop_c,0,0)$. Thus, from~\eqref{eqn:LMI:inter} we get that
\begin{align*}
r\dot{V}_o(e)+\dot{V}_c(\wh)\leq& -2r \delta V_o(e)-2 \kappa V_c(\wh) \\
\leq & -2\kappa\left( r  V_o(e)+ V_c(\wh) \right).
\end{align*} Therefore defining $V(\wh,e)=r  V_o(e)+ V_c(\wh)$, we get that
\[\dot{V}(\wh,e)\leq -2 \kappa V(\wh,e).\] Integrating in time,
\begin{equation}\label{eqn:LMI:inter}
r\ip{e}{\pop_o e}+\ip{\hat w}{\pinv_c \wh} \leq e^{-2 \kappa t}r\ip{w(0)}{\pop_o w(0)},
\end{equation} where we have used the fact that $\wh(0)=0$ and thus $e(0)=-w(0)$. Now, as discussed, there exist scalars $\theta_1,\theta_2>0$ such that
\[
\epsilon \norm{e}^2 \leq \ip{e}{\pop_o e} \leq \theta_1 \norm{e}^2,\quad
\frac{1}{\theta_2} \norm{\wh}^2 \leq \ip{\wh}{\pinv_c \wh} \leq \frac{1}{\epsilon} \norm{\wh}^2.
\]
Therefore, using~\eqref{eqn:LMI:inter} we get
\[\norm{e}^2 + \norm{\wh}^2 \leq \frac{r \theta_1}{\sigma}\norm{w(0)}^2 e^{-2 \kappa t},\]
where $\sigma=\min(r \epsilon,1/\theta_2)$. Thus,
\[
\norm{e},\norm{\wh} \leq \sqrt{\frac{r \theta_1}{\sigma}}\norm{w(0)} e^{- \kappa t}.\]
Finally, using the fact that $\norm{w}\leq \norm{\wh}+\norm{e}$ produces
\[
\norm{w(t)} \leq  2\sqrt{\frac{r \theta_1}{\sigma}}\norm{w(0)} e^{- \kappa t}.
\]

\end{proof}

The variables in Theorem~\ref{thm:LMI} are polynomials which are parameterized by vectors of coefficients associated to a predetermined monomial basis. There are two types of constraints on these variables: equality constraints between polynomials; and constraints of the form $\in \Xi_{d_1,d_2,\epsilon}$. To test the conditions of Theorem~\ref{thm:LMI}, these variables and constraints must be converted to a form recognized by an SDP solver such as SeDuMi~\cite{sturm1999using}. Many of these tasks have already been automated in SOSTOOLS~\cite{prajna2001introducing} and our extended toolbox, DelayTOOLS~\cite{peetlmi}. Specifically, SOSTOOLS has functionality for declaring polynomial variables and enforcing scalar equality constraints. Furthermore, DelayTOOLS~\cite{peetlmi} allows the user to declare matrix-valued equality constraints and create new polynomial variables which satisfy $\in \Xi_{d_1,d_2,\epsilon}$. Furthermore, the multipoly toolbox allows one to manipulate polynomial variables in order to construct new dependent polynomials such as in $\{\hat M,\hat K_1,\hat K_2\} =\Omega_c(M,K_1,K_2)$. Once all variables and constraints have been declared, SOSTOOLS converts all constraints and variables to a format which can be accepted by SDP solvers such as SeDuMi, SDPT3 or MOSEK. The a-posteriori polynomial manipulations such as operator inversion can be performed using a combination of the multipoly toolbox and Mupad. To help with understanding this process, we define several subroutines which perform specific relevant tasks and combine them in the pseudo code which would be used to obtain the observer-based controllers.
\begin{enumerate}
\item[] \hspace{-8mm} \texttt{[M,$\texttt{K}_\texttt{1}$,$\texttt{K}_\texttt{2}$]=mult\_semisep($\epsilon$)}
\begin{itemize}
 \item Declares polynomial variables $M$, $K_1$ and $K_2$ and enforces the constraint $\{M,K_1,K_2\} \in \Xi_{d_1,d_2,\epsilon}$.
 \end{itemize}
\item[] \hspace{-8mm} \texttt{[$\hat{\texttt{M}}$,$\hat{\texttt{K}}_\texttt{1}$,$\hat{\texttt{K}}_\texttt{2}$]=omega\_primal(M,$\texttt{K}_\texttt{1}$,$\texttt{K}_\texttt{2}$)}
\begin{itemize}
\item Constructs $\hat M$, $\hat K_1$ and $\hat K_2$ as defined by the map $\Omega_s$ in~\eqref{eqn:omega_s}.
\end{itemize}
\item[] \hspace{-8mm} \texttt{[$\hat{\texttt{M}}$,$\hat{\texttt{K}}_\texttt{1}$,$\hat{\texttt{K}}_\texttt{2}$]=omega\_dual(M,$\texttt{K}_\texttt{1}$,$\texttt{K}_\texttt{2}$)}
\begin{itemize}
\item Constructs $\hat M$, $\hat K_1$ and $\hat K_2$ as defined by the map $\Omega_c$ in~\eqref{eqn:omega_c}.
\end{itemize}
\item[] \hspace{-8mm} \texttt{eq\_constr(F)}
\begin{itemize}
\item Given a set $F$ of univariate/bivariate polynomials, declares element wise equality constraint $F=0$.
\end{itemize}
\item[] \hspace{-8mm} \texttt{[$\underbar{\texttt{M}}$,$\underbar{\texttt{K}}_{\texttt{1}}$,$\underbar{\texttt{K}}_{\texttt{2}}$]=inv\_op(M,$\texttt{K}_\texttt{1}$,$\texttt{K}_\texttt{2}$)}
\begin{itemize}
\item Given $\{M,K_1,K_2\}=\Xi_{d_1,d_2,\epsilon}$ calculates the inverse multiplier $\underbar M$ and kernels $\underbar K_1$ and $\underbar K_2$ by approximating $U(x)$ by performing the integration in~\eqref{eqn:inv:fundamental_matrix} a finite number of times and using a Chebyshev series approximation of $M(x)^{-1}$.
\end{itemize}
\item[] \hspace{-8mm} \texttt{[$\texttt{R}_{\texttt{1}}$,$\texttt{R}_{\texttt{2}}$]=controller\_gains(M,$\texttt{K}_\texttt{1}$,$\texttt{K}_\texttt{1}$,$\underbar{\texttt{M}}$,$\underbar{\texttt{K}}_{\texttt{1}}$,$\underbar{\texttt{K}}_{\texttt{2}}$)}
\begin{itemize}
\item The function defines $Y_1$ and $Y_2(x)$ using~\eqref{eqn:synth:Y}. Consequently, $R_1$ and $R_2(x)$ are defined using~\eqref{eqn:synth:R1} and~\eqref{eqn:synth:R2}, respectively.
\end{itemize}
\item[] \hspace{-8mm} \texttt{[$\texttt{L}_{\texttt{1}}$,$\texttt{L}_{\texttt{2}}$]=observer\_gains(M,$\texttt{K}_\texttt{1}$,$\texttt{K}_\texttt{1}$,$\underbar{\texttt{M}}$,$\underbar{\texttt{K}}_{\texttt{1}}$,$\underbar{\texttt{K}}_{\texttt{2}}$)}
\begin{itemize}
\item The function constructs $T_2$ using~\eqref{eqn:obs:T2} and sets $L_2$ using~\eqref{eqn:obs:gain:1}. Then the function constructs $T_1(x)$ and $T_3(x)$ using~\eqref{eqn:obs:T1} and~\eqref{eqn:obs:T3} and constructs $L_1(x)$ using~\eqref{eqn:obs:gain:2}.
\end{itemize}
\end{enumerate}
A pseudo code for the SOSTOOLS implementation of the SDP is presented in Algorithm~\ref{algo:output_feedback}.
\begin{algorithm}[]
\declareopt{\begin{enumerate}
\item \texttt{[M,$\texttt{K}_\texttt{1}$,$\texttt{K}_\texttt{2}$]=mult\_semisep($\epsilon$)}
\item \texttt{[N,$\texttt{S}_\texttt{1}$,$\texttt{S}_\texttt{2}$]=mult\_semisep($\epsilon$)}
\end{enumerate}}
\declarepoly{\begin{enumerate}
\item \texttt{[$\hat{\texttt{M}}$,$\hat{\texttt{K}}_\texttt{1}$,$\hat{\texttt{K}}_\texttt{2}$]=omega\_dual(M,$\texttt{K}_\texttt{1}$,$\texttt{K}_\texttt{2}$)}
\item \texttt{[$\hat{\texttt{N}}$,$\hat{\texttt{S}}_\texttt{1}$,$\hat{\texttt{S}}_\texttt{2}$]=omega\_primal(N,$\texttt{S}_\texttt{1}$,$\texttt{S}_\texttt{2}$)}
\end{enumerate}}
\declareconst{\begin{enumerate}
\item \texttt{eq\_constr((-$\hat{\texttt{M}}$-2$\mu$M,-$\hat{\texttt{K}}_{\texttt{1}}$-2$\mu \texttt{K}_{\texttt{1}}$,-$\hat{\texttt{K}}_{\texttt{2}}$-2$\mu \texttt{K}_{\texttt{2}}$) \\
-mult\_semisep(0))}
\item \texttt{eq\_constr((-$\hat{\texttt{N}}$-2$\delta$N,-$\hat{\texttt{S}}_{\texttt{1}}$-2$\delta \texttt{S}_{\texttt{1}}$,-$\hat{\texttt{S}}_{\texttt{2}}$-2$\delta \texttt{S}_{\texttt{2}}$)\\-mult\_semisep(0))}
\item \texttt{eq\_constr($\texttt{K}_{\texttt{2}}$(0,x))}
\item \texttt{eq\_constr($\texttt{S}_{\texttt{2}}$(0,x))}
\end{enumerate}}
\uIf{SOS problem is feasible}
{
    \outputs{$\texttt{M}$, $\texttt{K}_\texttt{1}$, $\texttt{K}_\texttt{2}$, $\texttt{N}$, $\texttt{S}_\texttt{1}$, $\texttt{S}_\texttt{2}$}.
}
\controlgains{
\begin{enumerate}
\item \texttt{[$\underbar{\texttt{M}}$,$\underbar{\texttt{K}}_{\texttt{1}}$,$\underbar{\texttt{K}}_{\texttt{2}}$]=inv\_op(M,$\texttt{K}_\texttt{1}$,$\texttt{K}_\texttt{2}$)}
\item \texttt{[$\texttt{R}_{\texttt{1}}$,$\texttt{R}_{\texttt{2}}$]=controller\_gains(M,$\texttt{K}_\texttt{1}$,$\texttt{K}_\texttt{1}$,$\underbar{\texttt{M}}$,$\underbar{\texttt{K}}_{\texttt{1}}$,$\underbar{\texttt{K}}_{\texttt{2}}$)}
\end{enumerate}}
\observergains{\begin{enumerate}
\item \texttt{[$\underbar{\texttt{N}}$,$\underbar{\texttt{S}}_{\texttt{1}}$,$\underbar{\texttt{S}}_{\texttt{2}}$]=inv\_op(N,$\texttt{S}_\texttt{1}$,$\texttt{S}_\texttt{2}$)}
\item \texttt{[$\texttt{L}_{\texttt{1}}$,$\texttt{L}_{\texttt{2}}$]=observer\_gains(N,$\texttt{S}_\texttt{1}$,$\texttt{S}_\texttt{1}$,$\underbar{\texttt{N}}$,$\underbar{\texttt{S}}_{\texttt{1}}$,$\underbar{\texttt{S}}_{\texttt{2}}$)}
\end{enumerate}}
\caption{Output-feedback controller synthesis.}
\label{algo:output_feedback}
\end{algorithm}


\section{Numerical Results}\label{sec:num_results}
In this section we test the conditions of Theorems~\ref{thm:analysis},~\ref{thm:synthesis} and~\ref{thm:observer} by applying them to two parameterized instances of scalar parabolic PDEs. The first instance is a variation of the classical isotropic heat equation. Because this system is well-studied, we are able to compare our results with a number of existing results in the literature. The second system is an anisotropic PDE with arbitrarily chosen coefficients. Both instances have an instability term, parameterized by an instability factor, $\lambda$. For both systems, we test stability, find controllers and construct observer-based controllers.

\emph{Example 1:} Our first system is defined as follows.
\begin{equation}\label{eqn:exmp1_PDE}
w_t(x,t)=w_{xx}(x,t)+\lambda w(x,t), \quad \lambda \in \R,
\end{equation} with boundary conditions
\[
w(0,t)=0, \quad w_x(1,t)=u(t).
\]
The output of the PDE is $v(t)=w(1,t)$. For $u(t)=0$, the analytical solution of this PDE is given by
\[
w(x,t)=\sum_{n=1}^\infty e^{\lambda_n t}\ip{w_0}{\phi_n}\phi_n(x),
\]
where $\lambda_n=\lambda - (2n-1)^2 \pi^2/4$ and $\phi_n=\sqrt{2}\sin ((2n-1)\pi x/2)$. This implies that Equation~\eqref{eqn:exmp1_PDE} is unstable for $\lambda>\pi^2/4 \approx 2.467$. To test the numerical accuracy of the stability conditions in Theorem~\ref{thm:analysis}, we found the largest $\lambda>0$ for which the conditions of Theorem~\ref{thm:analysis} are feasible as a function of the parameters $d_1$ and $d_2$ which define the degree of the variables $M$, $K_1$ and $K_2$. Table~\ref{table:exmp1:stab} presents these results for $\epsilon,\delta=0.001$. For $d_1=d_2=7$, we can construct a Lyapunov function which proves stability for $\lambda=2.461$, which is $99.74 \%$ of the stability margin $\frac{\pi^2}{4}\approx 2.4674$.

\begin{table}{}
\begin{center}
    \begin{tabular}{l *{8}{c}}\hline \hline
  $d=3$ & $4$ & $5$ & $6$ & $7$ & \text{analytic} \\ \hline
   $\lambda=0.59$ & $2.19$ & $2.457$ & $2.46$ & $2.461$ & $2.467$
\end{tabular}
\end{center}

\caption{Max. $\lambda$ as a function of $d_1=d_2=d$ for which the exp. stability conditions of Theorem~\ref{thm:analysis} are feasible, implying stability of PDE~\eqref{eqn:exmp1_PDE} with $u(t)=0$.}
\label{table:exmp1:stab}
\end{table}

To test the accuracy of the conditions in Theorem~\ref{thm:synthesis}, we find the largest $\lambda$ for which the conditions of Theorem~\ref{thm:synthesis} are feasible with $\epsilon=0.001$ and $\mu=0.001$, thereby implying the existence of an exponentially stabilizing state-feedback controller.  Table~\ref{table:exmp1:cont} presents this maximum $\lambda$ as a function of the degree $d_1=d_2=d$. The results suggest that for sufficiently high degree, a static state-feedback controller can be constructed for any value of $\lambda>0$.

\begin{table}{}
\begin{center}
    \begin{tabular}{l *{7}{c}}\hline \hline
   $d=7$ & $8$ & $9$ & $10$ & $11$ \\ \hline
   $\lambda=14.3982$ & $17.9626$ & $22.8645$ & $23.3093$ & $27.1179$
\end{tabular}
\end{center}

\caption{Max. $\lambda$ as a function of $d_1=d_2=d$ for which the conditions of Theorem~\ref{thm:synthesis} are feasible, thereby implying the existence of an exp. stabilizing state-feedback controller for PDE~\eqref{eqn:exmp1_PDE}.}
\label{table:exmp1:cont}
\end{table}

To test the accuracy of the conditions of Theorem~\ref{thm:observer}, we find the largest $\lambda$ for which the conditions of Theorem~\ref{thm:observer} are feasible with $\epsilon=0.001$ and $\delta=0.001$, thereby implying the existence of an exponentially stabilizing dynamic output-feedback controller with output $v(t)=w(1,t)$.
Table~\ref{table:exmp1:obs} presents this maximum $\lambda$ as a function of the degree $d_1=d_2=d$. The results suggest that for sufficiently high degree, a dynamic output feedback controller can be constructed for any value of $\lambda>0$.

\begin{table}{}
\begin{center}
    \begin{tabular}{l *{7}{c}}\hline \hline
  $d=7$ & $8$ & $9$ & $10$ & $11$ \\ \hline
   $\lambda = 14.5233$ & $17.7643$ & $23.4406$ & $24.7772$ & $27.8820$
\end{tabular}
\end{center}

\caption{Max. $\lambda$ as a function of $d_1=d_2=d$ for which the conditions of Theorem~\ref{thm:observer} are feasible, thereby implying the existence of an exp. stabilizing output-feedback controller for PDE~\eqref{eqn:exmp1_PDE}.}
\label{table:exmp1:obs}
\end{table}
\emph{Example 2:} To illustrate the versatility of the proposed method, we next consider the following arbitrarily chosen anisotropic system
\begin{equation}\label{eqn:exmp2_PDE}
w_t(x,t)=a(x)w_{xx}(x,t)+b(x)w_x(x,t)+c(x) w(x,t),
\end{equation}
where $a(x)=x^3-x^2+2$, $b(x)=3x^2-2x$ and $c(x)=-0.5x^3+1.3x^2-1.5x+0.7+\lambda$ with $\lambda \in \R$. Although the analytical solution to this PDE is not readily available, we may use a finite-difference scheme to numerically simulate the system and thereby estimate the range of $\lambda$ for which the PDE~\eqref{eqn:exmp2_PDE} is stable. Specifically, we find that the system is unstable for $\lambda>4.66$. To determine the accuracy of the conditions of Theorem~\ref{thm:analysis}, we find the largest $\lambda$ for which the conditions of Theorem~\ref{thm:analysis} are feasible. Table~\ref{table:exmp2:stab} lists the largest such $\lambda$ using $\epsilon,\delta=0.001$ as a function of polynomial degree $d_1=d_2=d$. The maximum $\lambda$ for which we can prove the exponential stability for is $\lambda=4.62$, which is $99.14 \%$ of the predicted stability margin of $4.66$. The $<1\%$ discrepancy may be due to conservatism or inaccuracy in the predicted maximum $\lambda$ on account of inaccuracy in the discretization or poor choice of initial conditions in the simulation.

\begin{table}{}
\begin{center}
    \begin{tabular}{l *{8}{c}}\hline \hline
  $d=3$ & $4$ & $5$ & $6$ & $7$ & \text{simulation} \\ \hline
  $\lambda=4.37$ & $4.61$ & $4.61$ & $4.62$ & $4.62$ & $4.66$
\end{tabular}
\end{center}

\caption{Max. $\lambda$ as a function of $d_1=d_2=d$ for which the exp. stability conditions of Theorem~\ref{thm:analysis} are feasible, implying stability of PDE~\eqref{eqn:exmp2_PDE} with $u(t)=0$.}
\label{table:exmp2:stab}
\end{table}
To test the accuracy of the conditions in Theorem~\ref{thm:synthesis}, we again find the largest $\lambda$ for which the conditions of Theorem~\ref{thm:synthesis} are feasible with $\epsilon=0.001$ and $\mu=0.001$, thereby implying the existence of an exponentially stabilizing state-feedback controller.  Table~\ref{table:exmp2:cont} presents this maximum $\lambda$ as a function of the degree $d_1=d_2=d$. The results suggest that for sufficiently high degree, a static state-feedback controller can be constructed for any value of $\lambda>0$.

\begin{table}{}
\begin{center}
    \begin{tabular}{l *{7}{c}}\hline \hline
  $d=4$ & $5$ & $6$ & $7$ & $8$  \\ \hline
  $\lambda=19.0216$ & $36.1359$ & $39.7247$ & $43.5974$ & $44.5219$
\end{tabular}
\end{center}

\caption{Max. $\lambda$ as a function of $d_1=d_2=d$ for which the conditions of Theorem~\ref{thm:synthesis} are feasible, thereby implying the existence of an exp. stabilizing state-feedback controller for PDE~\eqref{eqn:exmp2_PDE}.}
\label{table:exmp2:cont}
\end{table}
To test the accuracy of the conditions of Theorem~\ref{thm:observer}, we again find the largest $\lambda$ for which the conditions of Theorem~\ref{thm:observer} are feasible with $\epsilon=0.001$ and $\delta=0.001$, thereby implying the existence of an exponentially stabilizing dynamic output-feedback controller with output $v(t)=w(1,t)$.
Table~\ref{table:exmp2:obs} presents this maximum $\lambda$ as a function of the degree $d_1=d_2=d$. The results suggest that for sufficiently high degree, a dynamic output feedback controller can be constructed for any value of $\lambda>0$.

\begin{table}{}
\begin{center}
    \begin{tabular}{l *{7}{c}}\hline \hline
  $d=4$ & $5$ & $6$ & $7$ & $8$  \\ \hline
  $\lambda=18.3090$ & $36.0199$ & $38.0478$ & $40.5931$ & $44.079$
\end{tabular}
\end{center}

\caption{Max. $\lambda$ as a function of $d_1=d_2=d$ for which the conditions of Theorem~\ref{thm:observer} are feasible, thereby implying the existence of an exp. stabilizing output-feedback controller for PDE~\eqref{eqn:exmp2_PDE}.}
\label{table:exmp2:obs}
\end{table}

We conclude with the conjecture that the proposed method is asymptotically accurate in the sense that, for any $\lambda>0$, if the PDE~\eqref{eqn:prob:PDE_form}~-~\eqref{eqn:prob:PDE_form_BC} is stable in the autonomous sense, then the conditions of Theorem~\ref{thm:analysis} will be feasible for sufficiently high $d_1$ and $d_2$. Moreover, we conjecture that if the system is observable and controllable for some suitable definition of controllability and observability, then the conditions of Theorems~\ref{thm:synthesis} and~\ref{thm:observer} will be feasible for sufficiently high $d_1$ and $d_2$. We emphasize, however, that this is only a conjecture and additional work must be done in order to make this statement rigorous and determine its veracity. A further caveat to these results is the observation that the maximum degree $d_1$ and $d_2$ for which the conditions can be tested is a function of the memory and processing speed of the computational platform on which the experiments are performed. Specifically, the number of optimization variables in the underlying SDP problem is determined by the number of polynomial coefficients which scales as $O(d^2)$. To illustrate, all numerical experiments presented in this paper were performed on a machine with $8$ gigabytes of random access memory, which limited our analysis to a maximum degree of $d_1=d_2=11$ for PDE~\eqref{eqn:exmp1_PDE} and $d_1=d_2=8$ for PDE~\eqref{eqn:exmp2_PDE}.

In the following subsection, we illustrate the controllers and observers which result from feasibility of the conditions of Theorems~\ref{thm:synthesis} and~\ref{thm:observer} using numerical simulation.

\subsection{Numerical Implementation of Observer-Based Controllers}
To illustrate the observer-based controllers which result from feasibility of the conditions of Theorems~\ref{thm:synthesis} and~\ref{thm:observer}, we take the anisotropic PDE~\eqref{eqn:exmp2_PDE} with $\lambda=35$. This value of $\lambda$ renders the autonomous system unstable. We then synthesize controller and observer gains using the results of Theorems~\ref{thm:synthesis} and~\ref{thm:observer} for $d_1=d_2=6$, along with the inverse state transformation defined in Theorem~\ref{thm:invop}. For the inverse state transformation, $M(x)^{-1}$ is approximated using a sixth order Chebyshev series approximation and 5 iterations are used to define $U_\infty \cong U_5$. The controllers are then applied to the state and estimator dynamics, which are then discretized using a trapezoidal approximation. The initial state is set to
\[
w_0(x)=e^{-\frac{(x-0.3)^2}{2 (0.07)^2}}-e^{-\frac{(x-0.7)^2}{2 (0.07)^2}},
\]
while the initial observer state is set to $\hat w(x,0)=0$.
Figures~\ref{fig:controlled_state}~-~\ref{fig:control_effort} illustrate the state evolution of the system, observer and the control effort respectively. Finally, Figure~\ref{fig:gain} illustrates the integral control gain $R_2(x)$. Note that its behavior at the boundaries is logical since at $x=0$, the boundary condition $w(0,t)=0$ ensures that no control effort is required. Whereas, at $x=1$, the control exerts maximum effort.
\begin{figure*}[htp]
  \centering
  \subfigure[Evolution of closed-loop state $w(x,t)$.]{\includegraphics[scale=0.17]{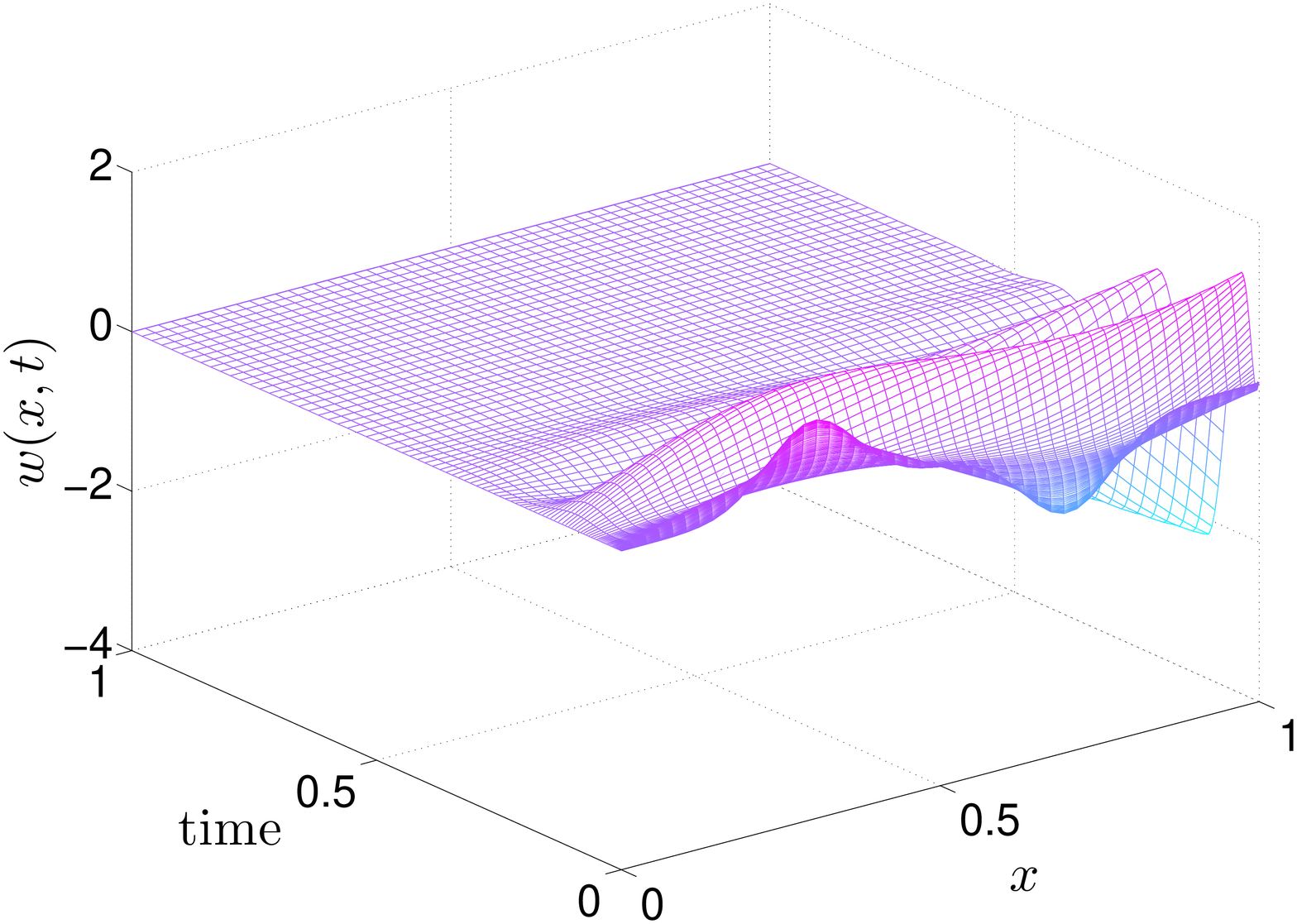}\label{fig:controlled_state}}
  \subfigure[Evolution of closed-loop state estimate $\hat w(x,t)$.]{\includegraphics[scale=0.17]{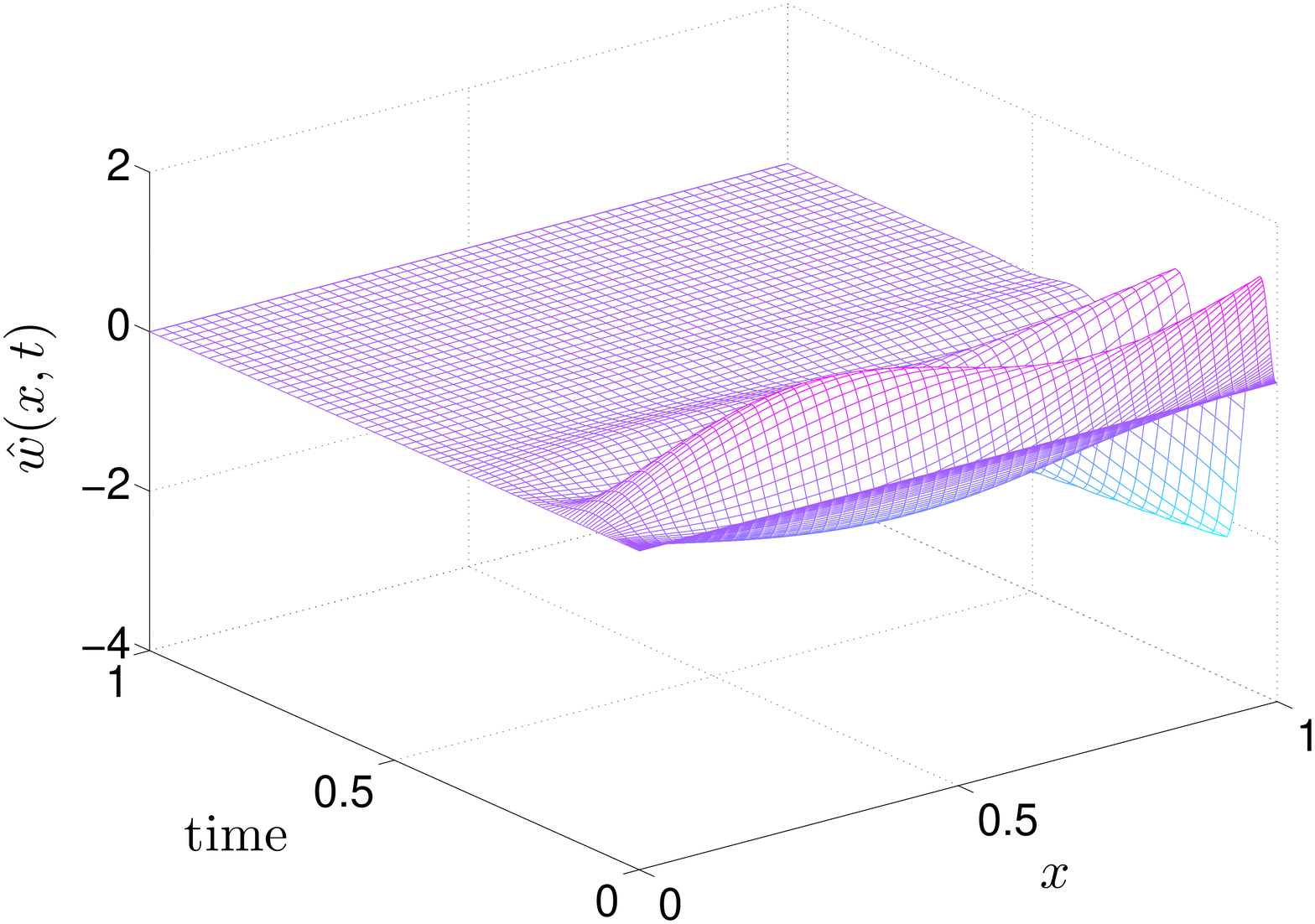}\label{fig:observer_state}}
  \subfigure[Control input $w_x(1,t)=u(t)$.]{\includegraphics[scale=0.17]{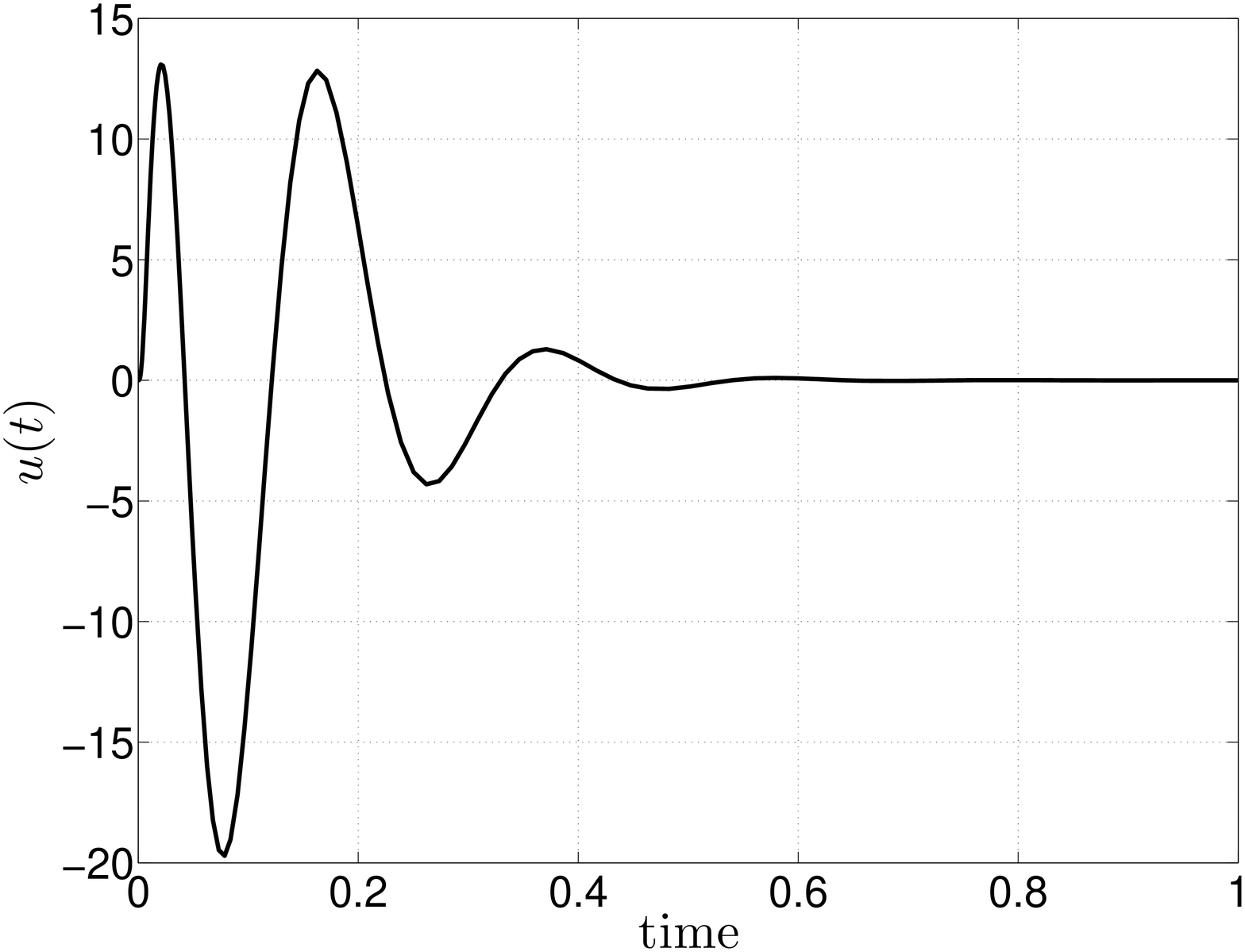}\label{fig:control_effort}}

  \caption{Evolution of closed loop system for Example 2 with $\lambda=35$ using controller from Theorem~\ref{thm:synthesis} and observer from Theorem~\ref{thm:observer}.}
\end{figure*}


\section{Necessity of Semi-Separable Kernels in the Lyapunov Function}\label{sec:simpler}
Recall that the Lyapunov functions used in Theorems~\ref{thm:analysis},~\ref{thm:synthesis}, and~\ref{thm:observer} all have the form
\begin{align*}
&V(w)= \int_0^1 w(x)M(x)w(x) dx \\
&+ \int_0^1 w(x) \bbbl( \int_0^x  K_1(x,\xi) w(\xi) d \xi  + \int_x^1 K_2(x,\xi) w(\xi) d \xi\bbbr) dx.
\end{align*}
As mentioned previously, this form is atypical in the study of parabolic PDEs and the reader may question the necessity of the terms $K_1$ and $K_2$ as their presence significantly complicates the analysis and increases the complexity of the stability conditions. Therefore, to illustrate the necessity of including these terms, in this section we repeat the numerical examples presented previously with the added restriction that $K_1=K_2=0$ (which translates to $P_{ij}=0$ for $i\neq j\neq 1$ in Theorem~\ref{thm:jointpos}). Table~\ref{table:simple_lyap1:cont} illustrates these results for the controller synthesis conditions of Theorem~\ref{thm:synthesis} using the same methodology as described in the previous section. These numerical tests indicate that while inclusion of $K_1$ and $K_2$ allows us to control the PDE for any $\lambda>0$, when $K_1=K_2=0$, our method will fail for some $\lambda$, regardless of the polynomial degree $d_1=d_1=d$. As indicated in Table~\ref{table:simple_lyap1:obs}, the results are similar for the observer synthesis conditions of Theorem~\ref{thm:observer}.
\begin{table}{}
\begin{center}
    \begin{tabular}{l *{8}{c}}\hline \hline
   & $d=1$ & $2$ & $3$ & $4\ldots 9$ & $10$ & \text{$K_1,K_2\neq 0$}\\ \hline
  Ex. $1$   & $\lambda=3.91$ & $4.78$  & $4.88$  & $4.88$ & $4.88$ & $27.1179$\\
  Ex. $2$ & $\lambda=3.51$ & $7.03$  & $8.59$  & $8.59$  & $8.59$ & $44.5219$
\end{tabular}
\end{center}

\caption{Re-evaluation of the results of Tables~\ref{table:exmp1:cont} and~\ref{table:exmp2:cont} with added constraint $K_1=K_2=0$.}
\label{table:simple_lyap1:cont}
\end{table}

\begin{table}{}
\begin{center}
    \begin{tabular}{l *{8}{c}}\hline \hline
   & $d=1$ & $2$ & $3$ & $4\ldots 9$ & $10$ & $K_1,K_2\neq0$\\ \hline
  Ex. $1$   & $\lambda=3.89$ & $4.79$  & $4.88$  & $4.88$ & $4.88$ & $27.8820$\\
  Ex. $2$ & $\lambda=3.51$ & $7.12$  & $8.43$  & $8.43$  & $8.43$ & $44.079$
\end{tabular}
\end{center}

\caption{Re-evaluation of the results of Tables~\ref{table:exmp1:obs} and~\ref{table:exmp2:obs} with added constraint $K_1=K_2=0$.}
\label{table:simple_lyap1:obs}
\end{table}


\section{Comparison With and Relation to Existing Results}\label{sec:comparison}
In this section, we compare our numerical results with several results in the literature which can be used for stability analysis and control, including those based on Sturm-Liouville theory and backstepping.

\subsection{Static Controllers Using Sturm Liouville Theory}\label{subsec:SL}
The output feedback controllers we construct are dynamic in that they rely on an auxiliary set of estimator dynamics which must be simulated in real-time. By contrast, static output feedback controllers do not use an estimator and instead rely only on a gain of the form, e.g. $u(t)=-\kappa v(t)=-\kappa w(1,t)$. Unfortunately, even for finite-dimensional systems the problem of static output feedback design is unsolved when $B \neq I$. That is, there is no LMI or polynomial-time algorithm which is guaranteed to find a stabilizing output feedback controller if one exists~\cite{syrmos1997static,fu2004pole}. However, there are numerous results which give sufficient conditions for the existence of such a controller, often based on the use of a fixed Lyapunov function. For the parabolic PDE which we consider, Sturm-Liouville theory~\cite[Chapter 2]{egorov1996spectral} can be used to express conditions for existence of static-output feedback controllers. Specfically, for $u(t)=-\kappa w(1,t)$, the stability of~\eqref{eqn:prob:PDE_form}~-~\eqref{eqn:prob:PDE_form_BC} depends on the first eigenvalue of the following Sturm-Liouville eigenvalue problem
\begin{equation}\label{eqn:SL:eigenproblem}
\frac{d}{dx}\left(p(x)\frac{d w(x)}{dx} \right)+q(x)w(x)=\mu \sigma(x)w(x),
\end{equation}
where $\mu$ is the eigenvalue and
\[
p(x)=e^{\int \frac{b(\xi)}{a(\xi)}d\xi}, \quad q(x)=c(x)\frac{p(x)}{a(x)}, \quad \sigma(x)=\frac{p(x)}{a(x)}.
\]
The boundary conditions for this eigenvalue problem are $w(0)=0$ and $w_x(1)+\kappa w(1)=0$. For our system, using the properties of the coefficients $a(x)$, $b(x)$ and $c(x)$ it can be established that $p$ is continuously differentiable, $q$ and $\sigma$ are continuous and there exist scalars $p_0$ and $\sigma_0$ such that $p(x) \geq p_0 >0$ and $\sigma(x) \geq \sigma_0 >0$. If $\mu_1$ is the first eigenvalue of~\eqref{eqn:SL:eigenproblem}, then it can be established using the Rayleigh quotient that $\mu_1 \leq \mu_1^{cc}$, where $\mu_1^{cc}$ is the first eigenvalue of the following constant coefficient Sturm-Liouville eigenvalue problem
\begin{equation}\label{eqn:SL:constant_eigenproblem}
p_0 \frac{d^2 w(x)}{w(x)}+q_1 w(x)=\mu^{cc}\sigma_1 w(x),
\end{equation}
subject to the boundary conditions $w(0)=0$ and $w_x(1)+\kappa w(1)=0$ and where $q_1$ and $\sigma_1$ are scalars such that
\[
q(x) \leq q_1 \quad \text{and} \quad \sigma(x) \leq \sigma_1.
\]
Now let us first consider Numerical Example 1, as defined in Equation~\eqref{eqn:exmp1_PDE} in Section~\ref{sec:num_results}. In this case, we have that $p_0=1$, $q_1=\lambda$ and $\sigma_1=1$. Therefore, estimating the first eigenvalue of~\eqref{eqn:SL:constant_eigenproblem} we get that $\mu_1^{cc} \approx \lambda- \pi^2$. Since, for stability we require $\mu_1^{cc} < 0$, for a large enough $\kappa>0$, a control input of the form $u(t)=-\kappa w(1,t)$ can stabilize~\eqref{eqn:exmp1_PDE} for $\lambda <  \pi^2$. This result is significantly more conservative than the results described in Tables~\ref{table:exmp1:cont}-\ref{table:exmp1:obs} which yield a stabilizing controller for at least $\lambda<27.1179$. Of course this result is not particularly surprising, as static output feedback controllers are a subset of dynamic output feedback controllers.

Similarly, for Numerical Example 2 (Equation~\eqref{eqn:exmp2_PDE}) we have $p(x)=x^3-x^2+2$, $q(x)=-0.5x^3+1.3 x^2-1.5x+0.7+\lambda$ and $\sigma(x)=1$.  Thus $p_0=50/27$, $q_1=0.7+\lambda$ and $\sigma_1=1$. Therefore, estimating the first eigenvalue of~\eqref{eqn:SL:constant_eigenproblem} we get that $\mu_1^{cc} \approx \lambda-17.58$. As before, we require $\mu_1^{cc} < 0$. Therefore for a large enough $\kappa>0$, a control input of the form $u(t)=-\kappa w(1,t)$ can stabilize~\eqref{eqn:exmp1_PDE} for $\lambda < 17.58$. Whereas, from Tables~\ref{table:exmp2:cont}-\ref{table:exmp2:obs} we see that Theorems~\ref{thm:synthesis} and~\ref{thm:observer} yield a dynamic output feedback controller for at least $\lambda<44.079$.

\subsection{The Case When $\mcl{A}+\mcl{A}^\star \leq 0$}\label{subsec:pos_real}
For some values of the coefficients $a(x)$, $b(x)$ and $c(x)$ we may have that $\mcl{A}+\mcl{A}^\star \leq 0$ on $\mcl{D}_0$, where the operator $\mcl{A}$ is defined in~\eqref{eqn:exist:A} and the set $\mcl{D}_0$ is defined in~\eqref{eqn:D0}. The output feedback stabilization of such systems, i.e. systems with $\mcl{A}+\mcl{A}^\star \leq 0$ and collocated control/observation, is considered in~\cite{curtain2006exponential}. The authors in~\cite{curtain2006exponential} show that for such systems there exists a scalar $\kappa >0$ (possibly $\kappa=\infty$) such that the control $u(t)=-\kappa v(t)$ exponentially stabilizes the system.  We wish to see if our methodology offers a performance gain over the controller proposed in~\cite{curtain2006exponential}. If we choose $a(x)=1$, $b(x)=0$ and $c(x)=\pi^2/4$, then
\begin{equation}\label{eqn:pos_real:A}
\mcl{A}=\frac{d^2}{dx^2}+\frac{\pi^2}{4}.
\end{equation} Applying integration by parts and Lemma~\ref{lem:wirtinger}, it can be established that $\mcl{A}+\mcl{A}^\star \leq 0$ on $\mcl{D}_0$. If we apply a controller of the form proposed in~\cite{curtain2006exponential}, then $u(t)=-\kappa v(t)=-\kappa w(1,t)$, for some $\kappa >0$. Using the theory in Subsection~\ref{subsec:SL} it is easily established that even for an arbitrarily large $\kappa > 0$, the closed loop system state will decay with a rate close to, but less then $3 \pi^2/4$. Whereas, from Table~\ref{table:pos_real} we observe that for $d_1=d_2=11$ we can construct an output feedback controller with a minimum exponential decay rate of $25.78$, a significant improvement over $3 \pi^2/4$.

\begin{table}{}
\begin{center}
    \begin{tabular}{l *{7}{c}}\hline \hline
    $d=6$ & $7$ & $8$ & $9$ & $10$ & $11$ \\ \hline
     $\delta=8.01$ & $12.7$  & $17.21$  & $20.31$ & $22.66$ & $25.78$
\end{tabular}
\end{center}

\caption{Max. exp. decay rate $\delta$ as a function of polynomial degree, $d_1=d_2=d$ for Equations~\eqref{eqn:prob:PDE_form}~-~\eqref{eqn:prob:PDE_form_BC} with $\mcl{A}$ as in Equation~\eqref{eqn:pos_real:A} for which we can construct output feedback controllers using Theorems~\ref{thm:synthesis} and~\ref{thm:observer}.}
\label{table:pos_real}
\end{table}

\begin{figure}[h]
\centering
\includegraphics[scale=0.2]{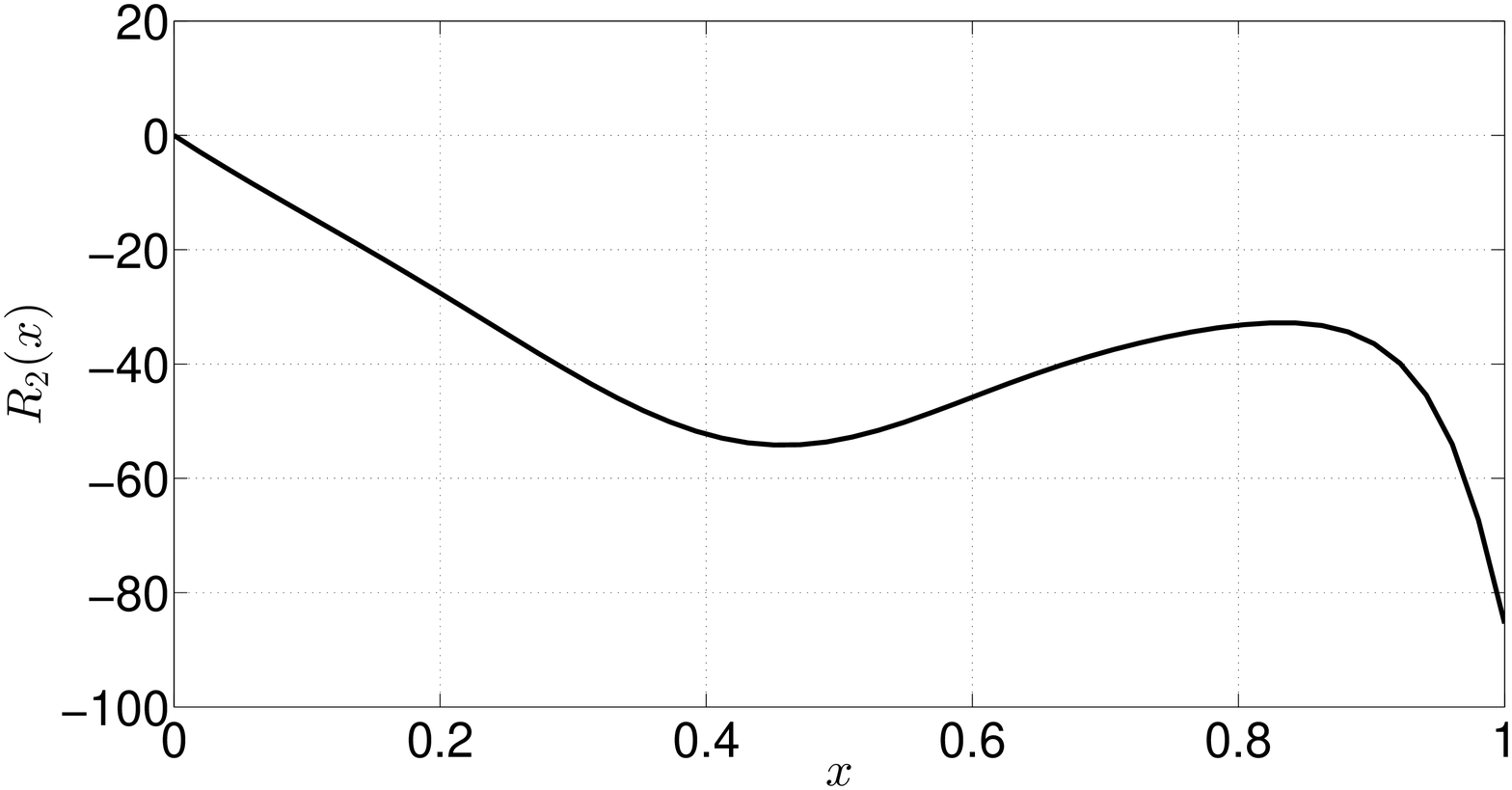}

\caption{ Control gain $R_2(x)$.}
\label{fig:gain}
\end{figure}

\subsection{Backstepping}\label{subsec:back}
Backstepping is a well-known alternative for the construction of stabilizing controllers for parabolic PDEs.
 Specifically, the backstepping approach defines a control law which, when coupled with an invertible state transformation, converts the controlled parabolic PDE to the form of a desired stable PDE (the target system). Although backstepping is not an optimization-based method and does not explicitly search for a Lyapunov-based stability proof, it turns out that the existence of a backstepping controller typically implies the existence of a Lyapunov function of the Form~\eqref{eqn:Lyapunov}, defined by a multiplier $M$ and semiseparable kernels $K_1$ and $K_2$. To demonstrate the existence of this Lyapunov function, let us consider the system defined by Example 1,
\begin{align}
&\label{eqn:back:system}w_t(x,t)= w_{xx}(x,t)+\lambda w(x,t),\\
&\label{eqn:back:system_BC}w(0,t)=0, \quad w_x(1,t)=u(t),
\end{align} where $\lambda > 0$.
Now define the \textit{target system}
\begin{align}
&\label{eqn:back:target}z_t(x,t)= z_{xx}(x,t),\\
&\label{eqn:back:target_BC}z(0,t)=0, \quad z_x(1,t)=0.
\end{align}
The key backstepping result is that there exists a function $E$ such that if
\[
u(t)=E(1,1)w(1,t)+\igzo (D_1 E)(1,x)w(x,t)dx,
\]
then for any solution $w$ of Equations~\eqref{eqn:back:system}~-~\eqref{eqn:back:system_BC},
\begin{equation}\label{eqn:back:transform}
z(x,t)=w(x,t)-\igzx E(x,\xi)w(\xi,t)d \xi, \nonumber
\end{equation}
is a solution of the target system in Equations~\eqref{eqn:back:target}~-~\eqref{eqn:back:target_BC}. Furthermore, if the map $\mcl{E}:w \mapsto z$ is invertible, then stability of the target system implies stability of the original closed-loop PDE. For the example problem given, this $E$ is obtained as a solution of a \textit{kernel-PDE} and can be found explicitly as~\cite{krstic2008boundary}
\begin{equation}\label{eqn:back:transform_kernel}
E(x,\xi)=-\lambda \xi \frac{I_1 \left(\sqrt{\lambda\left(x^2-\xi^2 \right)} \right)}{\sqrt{\lambda\left(x^2-\xi^2 \right)}}, \quad 0 \leq \xi \leq x \leq 1,
\end{equation} where $I_1$ is the first order modified Bessel function of the first kind. Moreover, $\mcl{E}$ has an inverse of the form
\begin{align}
\left(\mcl{E}^{-1}z \right)(x)&=z(x,t)+\igzx F (x,\xi)z(\xi,t)d \xi,\label{eqn:back:inverse_transform}
\end{align} where
\begin{equation}\label{eqn:back:inverse_transform_kernel}
F(x,\xi)=-\lambda \xi \frac{J_1 \left(\sqrt{\lambda\left(x^2-\xi^2 \right)} \right)}{\sqrt{\lambda\left(x^2-\xi^2 \right)}}, \quad 0 \leq \xi \leq x \leq 1,
\end{equation} where $J_1$ is the first order Bessel function of the first kind. Using properties of Bessel functions, it can be shown that both kernels $E$ and $F$ are bounded on the domain $\{(\xi,x)\,:\,0 \leq \xi \leq x \leq 1\}$. This implies that both $\mcl{E}$ and $\mcl{E}^{-1}$ are bounded with induced norms which we denote by $\norm{\mcl{E}}_{\mcl{L}}$ and $\norm{\mcl{E}^{-1}}_{\mcl{L}}$.

Now, to understand how this backstepping transformation implies the existence of a Lyapunov function with semi-separable kernels, we first note that stability of the target system in Equations~\eqref{eqn:back:target}~-~\eqref{eqn:back:target_BC} is established using the simple Lyapunov function
\begin{align}
V_{target}(z)=&\igzo z(x)^2dx=\ip{z}{z}, \nonumber
\end{align}
for which, using~\eqref{eqn:back:target}~-~\eqref{eqn:back:target_BC}, integration by parts and Lemma~\ref{lem:wirtinger}, we obtain
\begin{equation}\label{eqn:back:target_decay}
\frac{d}{dt} V_{target}(z(t))\le -\epsilon V_{target}(z(t)),
\end{equation} for any $z$ which satisfies~\eqref{eqn:back:target}~-~\eqref{eqn:back:target_BC}, where $\epsilon=\frac{\pi^2}{2}$. This implies
\[V_{target}(z(t))\le e^{-\epsilon t} V_{target}(z(0)) \Rightarrow \norm{z(x,t)} \leq e^{-\frac{\epsilon}{2}t}\norm{z(x,0)}.\]

Now, for the original system we define the Lyapunov function
\begin{align}
V_{plant}(w)=\ip{\mcl{E}w}{\mcl{E}w}.
\end{align}
Now, since for any solution, $w(t)$, of the original system, $z=\mcl{E}w(t)$ is a solution of the target system, we have that
\begin{align}
\frac{d}{dt} V_{plant}(w(t))&=\frac{d}{dt}\ip{\mcl{E}w(\cdot,t)}{\mcl{E}w(\cdot,t)} \notag \\
&=\frac{d}{dt}\ip{z(\cdot,t)}{z(\cdot,t)} \notag \\
&\label{eqn:back:inter1}=\frac{d}{dt}V_{target}(z(t))\le -\epsilon V_{target}(z(t))\notag \\
&=-\epsilon V_{target}(\mcl{E}w(t))=-\epsilon\ip{\mcl{E}w(t)}{\mcl{E}w(t)}\notag \\
& = -\epsilon V_{plant}(w(t)).\notag
\end{align}
Therefore,
\[V_{plant}(w(t)) \leq e^{-\epsilon t} V_{plant}(w(0)),\] which means
\begin{equation}\label{eqn:back:inter2}
\norm{\mcl{E}w(\cdot,t)} \leq e^{-\frac{\epsilon}{2}t} \norm{\mcl{E}w(\cdot,0)}.
\end{equation}
Boundedness of $\mcl{E}$ and $\mcl{E}^{-1}$ now implies $\norm{w(t)}\leq \norm{\mcl{E}^{-1}}_{\mcl{L}}\norm{\mcl{E}w(t)}$ and $\norm{\mcl{E}w(0)} \leq \norm{\mcl{E}}_{\mcl{L}}\norm{w(0)}$, which yields
\[
\norm{w(t)} \leq \norm{\mcl{E}^{-1}}_{\mcl{L}} \norm{\mcl{E}}_{\mcl{L}} e^{-\frac{\epsilon}{2}t} \norm{w(0)},\] which proves that $V_{plant}(w)=\norm{\mcl{E}w}^2$ establishes exponential stability of the original system.

We now show that $V_{plant}(w)$ has a form consistent with Theorem~\ref{thm:synthesis}. Expanding
\[V_{plant}(w)=\ip{\mcl{E}w}{\mcl{E}w},\]
we get
\begin{align*}
V_{plant}(w)=&\igzo w(x)^2-\igzo\igzx w(x) E(x,\xi)w(\xi)d\xi dx \\
&  -\igzo \igxo w(x) E(\xi,x)w(\xi)d\xi dx  \\
&+\igzo \igxo \int_0^\xi w(x)E(\xi,x)E(\xi,\eta)w(\eta)d\eta d\xi dx.
\end{align*} Changing the order of integration twice in the last integral and collecting like terms, we obtain
\begin{align*}
V_{plant}(w)=&\int_0^1 w(x)^2 dx \\
&\qquad  +\igzo \igzx w(x)H_1(x,\xi)w(\xi)d\xi dx  \\
& \qquad \qquad +   \igzo \igxo w(x)H_2(x,\xi)w(\xi)d\xi dx,\\
&=\ip{w}{\mcl{X}_{\{I,H_1,H_2\}}w},
\end{align*}
where
\begin{align*}
H_1(x,\xi)=&\igxo E(\eta,x)E(\eta,\xi)d \eta - E(x,\xi),\\
H_2(x,\xi)=&\int_\xi^1 E(\eta,x)E(\eta,\xi)d \eta - E(\xi,x),
\end{align*}
which has the form of a Lyapunov function consistent with Equation~\eqref{eqn:Lyapunov} using a semi-separable kernel where we have $M(x)=1$, $K_1=H_1$ and $K_2=H_2$. In a similar manner, if we define $\mcl{P}=\mcl{X}_{\{I,G_1,G_2\}}$ where
\begin{align*}
G_1(x,\xi)=&\igxo F(\eta,x)F(\eta,\xi)d \eta + F(x,\xi),\\
G_2(x,\xi)=&\int_\xi^1 F(\eta,x)F(\eta,\xi)d \eta + F(\xi,x),
\end{align*}
then $\mcl{P}^{-1}=\mcl{X}_{\{I,H_1,H_2\}}$ and hence
\begin{align}
V_{plant}(w)=\ip{\mathcal{P}^{-1} w}{\mcl{P}\mathcal{P}^{-1} w},
\end{align}
which is a form consistent with Theorem~\ref{thm:synthesis}. Thus we conclude that for this class of systems, if we assume the function $F$ may be approximated by polynomials, then the existence of a backstepping controller implies the feasibility of Theorem~\ref{thm:synthesis} for some degree. Despite this similarity, there are, of course, differences between the proposed method and backstepping. Specifically, our approach is optimization based, whereas the search for the backstepping transformation is not. Advantages of the proposed method include the ability to analyze stability of autonomous PDEs and simple extensions to robust control of PDEs with parametric uncertainty via Positivstellensatz results~\cite{putinar1993positive}.


\subsection{Finite-Dimensional Approximations}\label{subsec:LQR}
In this subsection we consider the merits of the SOS approach with respect to finite-dimensional approximation. That is, we consider whether there are advantages over model reduction techniques wherein the PDE is reduced to a set of coupled ODEs - as in, e.g.~\cite{balas1979feedback}.

Before continuing, we note that establishing a suitable metric for comparison of finite-dimensional and infinite-dimensional approaches is complicated by the fact that that the methods proposed in this paper are suboptimal. That is, we are not seeking observer-based controllers which are optimal in any sense. Rather, we simply seek observer-based controllers which establish closed-loop stability. In this sense, our methods are roughly equivalent to existing finite-dimensional approaches in that for all numerical examples considered, we are able to construct observer-based controllers for suitably high polynomial degree. In a sense, then, one could argue that finite-dimensional approaches are superior in that they are able to go beyond stabilization and construct \textit{optimal} observer-based controllers using a suitably high level of discretization. In practice, however, our experience has shown that there are disadvantages to discretization-based methods such as pole-placement. Specifically, we have seen that if the reduction scheme is not carefully chosen, discretization may result in loss of controllability or poorly conditioned controllability matrices. To illustrate, consider the following model:
\begin{align*}
&w_t(x,t)=w_{xx}(x,t)+15 w(x,t),\\
&w(0,t)=0, \quad w_x(1,t)=u(t).
\end{align*}
One approach to reduction of this PDE to a system of ODEs is to use a finite difference method to approximate the spatial derivative as
\begin{align*}w_{xx}(x,t)\approx& \frac{2}{\Delta x_1+\Delta x_2}\left(\frac{w(x+\Delta x_2,t)-w(x,t)}{\Delta x_2}   \right)\\
&-\frac{2}{\Delta x_1+\Delta x_2} \left(\frac{w(x,t)-w(x-\Delta x_1,t)}{\Delta x_1}   \right)  ,
\end{align*} where $\Delta x_1$ is the step size to the left of $x$ and $\Delta x_2$ is the step size to the right. Using this scheme we obtain an ODE model of the form
\begin{equation}\label{eqn:LQR:reduced_model}
\dot{w}^m(t)=A^m w^m(t) + B^m u(t),
\end{equation} where $w^m(t),B^m \in \R^{m \times 1}$ and $A^m \in \R^{m \times m}$ and $m \in \N$ is the order of reduction. While relatively straightforward, this approach creates significant technical challenges. For example:

\textit{a) Controllability of the Reduced Model:} The reduced-order model must be chosen so as to maintain the properties of controllability and observability. In most cases, however, there is no guarantee that a finite-difference approximation scheme will preserve these properties. For example, for the finite difference scheme defined above, it is known that if the original system is controllable and a uniform grid size is chosen, then the reduced system is also controllable. However, if one were to chose a non-uniform grid, then controllability is no longer guaranteed. For example if one were to chose a logarithmic grid, for $m>13$ the reduced model is not controllable (although it is still  stabilizable). In such a case, the performance of the closed loop system will be limited by the location of the uncontrollable eigenvalues.

\textit{b) Ill-conditioned Controllability Matrix:} Now suppose we wish to perform pole placement by applying Ackermann's formula to the reduced order model. As mentioned, it can be shown that the reduced order model in~\eqref{eqn:LQR:reduced_model} is controllable  for any $m \in \N$ when derived using uniform step sizes ($\Delta x_1=\Delta x_2$)  as established by the Hautus test. However, the pole placement problem (which is similar to our condition for exponential stabilization with desired decay rate) relies on inversion of the controllability matrix $\mcl{C}(A^m,B^m)$ - a step which is numerically sensitive to conditioning of $\mcl{C}(A^m,B^m)$. This is problematic since, as seen in Table~\ref{table:controllability_condition}, the controllability matrix for this system is ill-conditioned and the condition number \textit{worsens} as the level of disretization $m$ \textit{increases}. This implies that as the level of discretization increases, numerical errors may dominate - potentially resulting in unstable or unpredictable controllers. Naturally, these issues are well-known and have been addressed in the literature through methods such as robust place placement~\cite{tits1996globally} or Galerkin schemes~\cite{kunisch2001galerkin}. The advantage of the SOS approach, however, is that the controllers are provably stable at the pre-lumping stage and thus the only numerical concern is implementation, which does not appear to be sensitive to issues such as condition number.
\begin{table}{}
\begin{center}
    \begin{tabular}{l *{7}{c}}\hline \hline
  $m$  & $5$ & $10$ & $20$ \\ \hline
  $cond(\mcl{C}(A^m,B^m))\approx$ & $ 10^7$ & $10^{25}$ & $10^{63}$
\end{tabular}
\end{center}

\caption{Condition number of $\mcl{C}(A^m,B^m)$ as a function of order of reduction $m$.}
\label{table:controllability_condition}
\end{table}

\section{Alternative Boundary Conditions}\label{sec:ABC}
The results of this paper may be readily adapted to other types of boundary conditions. Specifically, the conditions of Theorems~\ref{thm:analysis}, ~\ref{thm:synthesis} and~\ref{thm:observer} can be easily modified to consider alternative boundary conditions. Although economy of space prohibits us from presenting these conditions in full, in this section we give the results of numerical tests performed using Dirichlet, Neumann and Robin boundary conditions. Specifically, for the two PDEs~\eqref{eqn:exmp1_PDE} and~\eqref{eqn:exmp2_PDE} which define Examples~$1$ and~$2$, respectively, in Section~\ref{sec:num_results}, we consider the boundary conditions and the outputs as listed in Table~\ref{table:alt_BC}.
\begin{table}{}
\begin{center}
    \begin{tabular}{l *{3}{c}}\hline \hline
  & Boundary Condition & Output $v(t)$  \\ \hline
   Dirichlet & {$\!\begin{aligned}
               w(0,t) &= 0 \\
               w(1,t) &= u(t) \end{aligned}$} & $w_x(1,t)$  \\ \hline
   Neumann & {$\!\begin{aligned}
               w_x(0,t) &= 0 \\
               w_x(1,t) &= u(t) \end{aligned}$} & $w(1,t)$ \\ \hline
   Robin & {$\!\begin{aligned}
               w(0,t)+w_x(0,t) &= 0 \\
               w(1,t)+w_x(1,t) &= u(t) \end{aligned}$} & $w(1,t)$
\end{tabular}
\end{center}

\caption{Alternative boundary conditions and outputs for PDEs~\eqref{eqn:exmp1_PDE} and~\eqref{eqn:exmp2_PDE}.}
\label{table:alt_BC}
\end{table}

Tables~\ref{table:alt:exmp1:output} and~\ref{table:alt:exmp2:output} illustrate the maximum $\lambda$ for which we can construct output-feedback based controllers as a function of $d_1=d_2=d$ for PDEs~\eqref{eqn:exmp1_PDE} and~\eqref{eqn:exmp2_PDE}, respectively, for the boundary conditions listed in Table~\ref{table:alt_BC} using exponential decay rates of $\delta=\mu=0.001$.
\begin{table}{}
\begin{center}
    \begin{tabular}{l *{7}{c}}\hline \hline
    & $d=8$ & $9$ & $10$ & $11$ \\ \hline
  Dirichlet  & $\lambda=17.7634$ & $22.8645$ & $23.3093$ & $27.1179$ \\
  Neumann & $14.8163$ & $17.1814$ & $21.8781$ & $21.8781$  \\
  Robin & $13.8367$ & $16.6565$ & $18.6050$ & $18.9758$\\
\end{tabular}
\end{center}

\caption{Max. $\lambda$ as a function of polynomial degree, $d_1=d_2=d$ for PDE~\eqref{eqn:exmp1_PDE} with boundary conditions as in Table~\ref{table:alt_BC} for which we can construct output-feedback boundary controllers.}
\label{table:alt:exmp1:output}
\end{table}
\begin{table}{}
\begin{center}
    \begin{tabular}{l *{7}{c}}\hline \hline
  & $d=5$ & $6$ & $7$ & $8$  \\ \hline
  Dirichlet & $\lambda=36.0199$ & $38.0478$ & $40.5930$ & $44.079$   \\
  Neumann & $29.8492$ & $31.1447$ & $31.1447$ & $34.1584$ \\
  Robin & $24.6490$ & $27.8503$ & $27.8503$ & $29.4373$ \\
\end{tabular}
\end{center}

\caption{Max. $\lambda$ as a function of polynomial degree, $d_1=d_2=d$ for PDE~\eqref{eqn:exmp2_PDE} with boundary conditions as in Table~\ref{table:alt_BC} for which we can construct output-feedback boundary controllers.}
\label{table:alt:exmp2:output}
\end{table}
Similar to the observation made in Section~\ref{sec:num_results}, the numerical results in this section suggest that our methodology is asymptotically accurate for the considered alternative boundary conditions, that is, given any $\lambda>0$, we can construct controllers/observers by choosing a large enough $d_1=d_2=d$. A more detailed study of alternative boundary conditions can be found in the thesis work of~\cite{Gahlawatthesis}.

 \section{Conclusion and Future Work}\label{sec:conclusion}
We have defined an algorithmic, polynomial-time approach to the design of observer-based controllers for a general class of scalar parabolic partial differential equations using measurements and feedback at the boundary. The results use polynomials and semidefinite programming to parameterize a convex set of positive Lyapunov functions on the Hilbert space $L_2$. By combining these Lyapunov functions with an invertible state transformation, we obtain convex conditions for stability, controller synthesis and Luenberger observer design. Furthermore, we have tested our results using parameterized numerical examples in order to show that the stability conditions are accurate to several significant figures and the synthesis conditions yield controllers for a large class of controllable and observable systems. Furthermore, we have adapted the approach to three alternative classes of boundary measurements and actuators. Finally, we have performed a series of comparisons with existing results in the literature, showing, e.g. that the method is analytically equivalent to backstepping for controller synthesis and furthermore is numerically competitive for the examples considered. By using an optimization-based algorithm defined by polynomials, the results presented here have the advantage that they may be further extended to the problem of nonlinear stability analysis, robust control, and control of coupled, multivariate, hyperbolic and elliptic PDEs - topics of ongoing research.



\appendix\label{sec:appendix}

To facilitate presentation in this appendix, we use the following lemmas. The first is simply a restatement of the Wirtinger inequality
\begin{lemma}[\cite{seuret2012use}]
\label{lem:wirtinger}
 Let $z \in H^2(0,1)$ be a scalar function. Then
  \[\int_0^1 (z(x)-z(0))^2dx \leq \frac{4}{\pi^2} \int_0^1 z_x(x)^2 dx.
  \]
\end{lemma}
The second lemma is accomplished by splitting the integral in two parts and applying a change in the variable of integration to the second part.
\begin{lemma}\label{lem:order}
For any bivariate polynomials $K$ and $P$ the following identity holds for any $w \in \lt$
\begin{align*}
&\igzo w(x)\left(\igzx K(x,\xi)w(\xi)d\xi + \igxo P(x,\xi)w(\xi)d\xi\right)dx \\
&= \igzo \igzx w(x) \hlf \left[K(x,\xi)+P(\xi,x) \right] w(\xi) d\xi dx \\
&\qquad + \igzo \igxo w(x) \hlf \left[P(x,\xi)+K(\xi,x) \right] w(\xi) d\xi dx.
\end{align*}
\end{lemma}
\begin{lemma}[Analysis]\label{lem:analysis} \label{lem:appendix_1}
Given polynomials $a$, $b$ and $c$ with $a(x) \geq \alpha >0$, for all $x \in [0,1]$, suppose that there exists a scalar $\epsilon>0$ and polynomials $M$, $K_1$ and $K_2$ such that
\begin{align*}
&\{M,K_1,K_2\} \in \Xi_{d_1,d_2,\epsilon},\\
& (b(1)-a_x(1))K_1(1,x)-a(1)(D_1 K_{1})(1,x)=0,\\
& (b(1)-a_x(1))M(1)-a(1)M_x(1) \leq 0,\\
& K_2(0,x)=0.
\end{align*}
Let
\begin{align*}
&V(w)=\ip{w}{\mcl{X}_{\{M,K_1,K_2\}}w},
\end{align*} where $\mcl{X}_{\{M,K_1,K_2\}}$ is as defined in~\eqref{eqn:X}.
Then, for any $w$ which satisfies Equations~\eqref{eqn:stab:PDE_form}~-~\eqref{eqn:stab:PDE_form_BC},
\begin{align*}
&\dot V(w(t))\le
\ip{w(t)}{\mcl{X}_{\{\hat{M},\hat{K}_1,\hat{K}_2\}}w(t)},
\end{align*}
where $\{\hat{M},\hat{K}_1,\hat{K}_2\}=\Omega_s \{M,K_1,K_2\}$.
\end{lemma}
\begin{proof}
Let $\mcl{P}=\mcl{X}_{\{M,K_1,K_2\}}$ so that $V(w)=\ip{w}{\pop w}$. If $w$ satisfies~\eqref{eqn:stab:PDE_form}~-~\eqref{eqn:stab:PDE_form_BC}, then taking the time derivative of $V(w(t))$ and since $\{M,K_1,K_2\} \in \Xi_{d_1,d_2,\epsilon}$ implies $\pop$ is self-adjoint, we can write
$\dot{V}(w(t))=2\ip{w_t}{\pop w}$. Using Equation~\eqref{eqn:stab:PDE_form} we expand this out to get
\begin{equation}\label{eqn:Lyapunov_derv:1}
\dot{V}(w(t))=2\ip{w_t}{\pop w}=2\sum_{n=1}^5 \Gamma_n,
\end{equation} where
 \begin{align*}
 \Gamma_1 =& \igzo w_{xx}(x,t)a(x)M(x)w(x,t)dx, \\
 \Gamma_2 =& \igzo w_x(x,t)b(x)M(x)w(x,t)dx, \\
 \Gamma_3 =& \sum_{i=1}^2 \int_{\Delta_i} w_{xx}(x,t)a(x)  K_i(x,\xi)w(\xi,t) d\xi dx,\\
 \Gamma_4 = &\sum_{i=1}^2 \int_{\Delta_i} w_x(x,t)b(x)  K_i(x,\xi)w(\xi,t) d\xi dx\\
 \Gamma_5 = & \igzo w(x,t)^2 M(x) c(x) dx \\
 &\qquad +  \sum_{i=1}^2 \int_{\Delta_i} w(x,t) c(x) K_i(x,\xi)w(\xi,t) d\xi dx,
\end{align*} where $\Delta_1=\{(\xi, x) \,:\, 0 \le \xi \le x \le 1\}$ and $\Delta_2=\{(\xi, x) \,:\, 0 \le x \le \xi \le 1\}$.
Applying integration by parts twice and using the boundary condition $w(0,t)=w_x(1,t)=0$ yields
 \begin{align*}
 \Gamma_1 =& - \igzo w_x(x,t)^2 a(x)M(x)dx  \\
 &\qquad + \hlf \igzo  \frac{\partial^2}{\partial x^2} \left[ a(x)M(x)\right] w(x,t)^2 dx \\
 &\qquad \qquad -  \frac{1}{2}\left(a_x(1)M(1)+ a(1)M_x(1) \right) w(1,t)^2 .
 \end{align*}
 Since $a(x) \geq \alpha >0$ and $\{M,K_1,K_2\} \in \Xi_{d_1,d_2,\epsilon}$, we have $a(x)M(x) \geq \alpha \epsilon$. Thus, by application of Lemma~\ref{lem:wirtinger} we get
 \[
 - \igzo w_x(x,t)^2 a(x)M(x) dx \leq -\frac{\pi^2}{4}\alpha \epsilon \igzo w(x,t)^2 dx.
 \]
Therefore, we conclude that
 \begin{align}
 \Gamma_1 \leq & \hlf \igzo w(x,t)^2\left( \frac{\partial^2}{\partial x^2} \left[ a(x)M(x)\right]  -\frac{\pi^2}{2} \alpha \epsilon \right) dx \notag \\
 &\label{eqn:Gamma_1}\qquad - \frac{1}{2}\left(a_x(1)M(1)+ a(1)M_x(1) \right) w(1,t)^2 .
\end{align}
Again, applying integration by parts once and using $w(0,t)=0$,
 \begin{align}
 \Gamma_2 \hspace{-1mm} = \hspace{-1mm}&\label{eqn:Gamma_2}  - \hspace{-1mm} \hlf \hspace{-1mm} \igzo \hspace{-1mm} w(x,t)^2   \pfx \left[ b(x)M(x)\right] dx \hspace{-1mm}  + \hspace{-1mm}   \frac{1}{2}b(1)M(1) w(1,t)^2.
 \end{align}
Since $\{M,K_1,K_2\} \in \Xi_{d_1,d_2,\epsilon}$, we have $K_1(x,\xi)=K_2(\xi,x)$ and thus $K_1(x,x)=K_2(x,x)$. Exploiting this property, the constraint $K_2(0,x)=0$, and the boundary conditions $w(0,t)=w_x(1,t)=0$, we apply integration by parts twice to obtain
\begin{align}
 \Gamma_3  = & \igzo \hspace{-1mm} w(x,t)^2\left( \left[ \pfx \left[a(x)(K_1(x,\xi) \hspace{-1mm} - \hspace{-1mm}K_2(x,\xi)) \right] \right]_{\xi=x} \right) dx \notag \\
 &  +  \sum_{i=1}^2 \int_{\Delta_i} w(x,t)  \left( \frac{\partial^2}{\partial x^2}\left[a(x)K_i(x,\xi) \right] \right) w(\xi,t) d\xi dx \nonumber \\
  &     - w(1,t) \igzo a_x(1)K_1(1,x) w(x,t) dx \nonumber \\
  &  -w(1,t) \igzo a(1)(D_1 K_{1})(1,x) w(x,t)dx. \nonumber
 \end{align}
   Applying Lemma~\ref{lem:order} and using $K_1(x,\xi)=K_2(\xi,x)$, we get
 \begin{align}
 \Gamma_3 \hspace{-1mm} &= \hspace{-1mm} \igzo w(x,t)^2\left( \left[ \pfx \left[a(x)(K_1(x,\xi)-K_2(x,\xi)) \right] \right]_{\xi=x} \right) dx  \notag \\
 & - \hspace{-1mm} w(1,t)  \hspace{-1mm} \igzo \hspace{-1mm} \left(a_x(1)K_1(1,x) \hspace{-1mm} + \hspace{-1mm} a(1)(D_1 K_{1})(1,x) \right) w(x,t) dx \nonumber \\
   &\label{eqn:Gamma_3}  +\sum_{i=1}^2 \hlf \int_{\Delta_i} w(x,t) \bmat{\frac{\partial^2}{\partial x^2}  \\ \frac{\partial^2}{\partial \xi^2}}^T \bmat{a(x)K_i(x,\xi) \\ a(\xi)K_i(x,\xi)}w(\xi,t)d\xi dx.
 \end{align}
  Applying integration by parts once and following the same procedure as for $\Gamma_3$, we get
 \begin{align}
 \Gamma_4 =&- \sum_{i=1}^2 \hlf  \int_{\Delta_i} w(x,t) \bmat{\frac{\partial}{\partial x} \\ \frac{\partial}{\partial \xi}}^T \bmat{b(x)K_i(x,\xi) \\ b(\xi)K_i(x,\xi)}w(\xi,t)d\xi dx \notag \\
  &\label{eqn:Gamma_4} \qquad \qquad + w(1,t) \igzo b(1)K_1(1,x)w(x,t)dx.
 \end{align}
Finally, employing Lemma~\ref{lem:order} produces
 \begin{align}
\Gamma_5 &=  \igzo w(x,t)^2 M(x) c(x) dx  \notag \\
 &\label{eqn:Gamma_5} + \sum_{i=1}^2 \hlf \int_{\Delta_i} w(x,t) \left( \left[c(x)+c(\xi) \right] K_i(x,\xi) \right) w(\xi,t) d\xi dx.
 \end{align}
Finally, we combine the terms~\eqref{eqn:Gamma_1}~-~\eqref{eqn:Gamma_5} into the derivative~\eqref{eqn:Lyapunov_derv:1} and use the constraints
 \begin{align*}
& (b(1)-a_x(1))K_1(1,x)-a(1)(D_1 K_{1})(1,x)=0,\\
& (b(1)-a_x(1))M(1)-a(1)M_x(1) \leq 0,
\end{align*} to eliminate extraneous terms, thereby completing the proof.
\end{proof}

\begin{lemma}[Controller Synthesis]\label{lem:control}
Given polynomials $a$, $b$ and $c$ with $a(x) \geq \alpha >0$, for all $x \in [0,1]$, suppose that there exists a scalar $\epsilon>0$ and polynomials $M$, $K_1$ and $K_2$ such that
\begin{align*}
&\{M,K_1,K_2\} \in \Xi_{d_1,d_2,\epsilon},\quad K_2(0,x)=0.
\end{align*} Let
\begin{align*}
V(w)=&\ip{w}{\pinv w},
\end{align*} where $\mcl{P}=\mcl{X}_{\{M,K_1,K_2\}}$ and $\mcl{X}_{\{M,K_1,K_2\}}$ is as defined in~\eqref{eqn:X}. Then, for any $w$ which satisfies~\eqref{eqn:synth:PDE_form}~-~\eqref{eqn:synth:PDE_form_BC}
\begin{align*}
\dot V(w(t))\le& \ip{y(t)}{\mcl{X}_{\{\hat{M},\hat{K}_1,\hat{K}_2\}}y(t)} \\
& \qquad +\left[a(1)M_x(1)+(b(1)-a_x(1))M(1)  \right]y(1,t)^2 \nonumber \\
&  \qquad \qquad +2a(1)M(1)y_x(1,t)y(1,t),
\end{align*}
where $y=\pinv w$ and $\{\hat{M},\hat{K}_1,\hat{K}_2\} \in \Omega_c \{M,K_1,K_2\}$.
\end{lemma}
\begin{proof}
Taking the time derivative of $V(w(t))$ and since $\pinv$ is self-adjoint, we obtain
\begin{align}
\dot{V}(w(t)) &= 2 \ip{w_t}{\pinv w} \notag \\
&= 2 \ip{a(\cdot)w_{xx}+b(\cdot)w_x+c(\cdot)w}{\pinv w} \nonumber \\
& = 2 \ip{a(\cdot)\frac{\partial^2}{\partial x^2}(\pop y)+b(\cdot)\frac{\partial}{\partial x}(\pop y)+c(\cdot) \pop y}{y}  \notag \\
&\label{eqn:lem2:Vdot1}= 2 \sum_{n=1}^5 \Gamma_n,
\end{align}
where $y=\pinv w$ and
\begin{align*}
\Gamma_1=& \igzo a(x) \frac{\partial^2}{\partial x^2}(M(x)y(x,t))y(x,t)dx, \\
\Gamma_2=& \igzo b(x) \frac{\partial}{\partial x}(M(x)y(x,t))y(x,t)dx, \\
\Gamma_3=& \sum_{i=1}^2 \igzo a(x) \frac{\partial^2}{\partial x^2} \left(\int_{\beta_i} K_i(x,\xi)y(\xi,t)d\xi \right)y(x,t)dx,\\
\Gamma_4=&\sum_{i=1}^2 \igzo b(x) \frac{\partial}{\partial x} \left(\int_{\beta_i} K_i(x,\xi)y(\xi,t)d\xi \right)y(x,t)dx ,\\
\Gamma_5=& \igzo c(x)M(x)y(x,t)^2 dx  \\
&\qquad + \sum_{i=1}^2\int_{\Delta_i} y(x,t)c(x)K_i(x,\xi)y(\xi,t)d\xi dx,
\end{align*} where  $\Delta_1=\{(\xi,x)\,:\,0 \le \xi \le x \le 1\}$, $\Delta_2=\{(\xi,x)\,:\,0 \le x \le \xi \le 1\}$, $\beta_1=[0,x]$ and $\beta_2=[x,1]$.
 Before proceeding we calculate $y(0,t)$. The definition $y=\pinv w$ implies
\[
w(0,t)=M(0)y(0,t)+\igzo K_2(0,x)y(x,t)dx.
\]
 Therefore, since $w(0,t)=0$ and $K_2(0,x)=0$, we get $y(0,t)=0$. Now, since $M(x)a(x) \geq \alpha \epsilon$ and $y(0,t)=0$, applying integration by parts twice and using Lemma~\ref{lem:wirtinger} produces
\begin{align}
\Gamma_1 \hspace{-1mm} \leq  & \hlf \igzo  \hspace{-1mm} \left( a_{xx}(x)M(x) \hspace{-1mm} + \hspace{-1mm} a(x)M_{xx}(x) \hspace{-1mm} - \hspace{-1mm} \frac{\pi^2}{2}\alpha \epsilon   \right)y(x,t)^2 dx  \notag \\
& \qquad \qquad + \hlf \bl(a(1)M_x(1)-a_x(1)M(1) \br)y(1,t)^2 \nonumber \\
& \qquad \qquad \qquad \qquad \label{eqn:lem2:Gamma1} \quad +a(1)M(1)y_x(1,t)y(1,t).
\end{align} Similarly, applying integration by parts once yields
\begin{align}
\Gamma_2=& \hlf \igzo \left(b(x)M_x(x)-b_x(x)M(x) \right)y(x,t)^2 dx\notag \\
&\label{eqn:lem2:Gamma2}\qquad \qquad \qquad \qquad  +\hlf b(1)M(1)y(1,t)^2.
\end{align}
Applying integration by parts twice and Lemma~\ref{lem:order} yields
\begin{align}
\Gamma_3=& \igzo \left(a(x) \left[\frac{\partial}{\partial x}\left[K_1(x,\xi)-K_2(x,\xi) \right]  \right]_{\xi=x}  \right)y(x,t)^2 dx \notag \\
&\label{eqn:lem2:Gamma3}+  \sum_{i=1}^2 \hlf \int_{\Delta_i} y(x,t)\bmat{a(x) \frac{\partial^2}{\partial x^2} \\ a(\xi) \frac{\partial^2}{\partial \xi^2}}^T \bmat{K_i(x,\xi)\\ K_i(x,\xi)}  y(\xi,t)d\xi dx.
\end{align} In a similar manner as for $\Gamma_3$, we obtain
\begin{align}
\Gamma_4=&\label{eqn:lem2:Gamma4} \sum_{i=1}^2 \hlf \int_{\Delta_i} y(x,t)\bmat{b(x) \frac{\partial}{\partial x} \\ b(\xi) \frac{\partial}{\partial \xi}}^T \bmat{K_i(x,\xi)\\ K_i(x,\xi)}y(\xi,t)d\xi dx.
\end{align} Finally, applying Lemma~\ref{lem:order} to $\Gamma_5$ produces
\begin{align}
\Gamma_5=&\igzo c(x)M(x)y(x,t)^2 dx \notag \\
&\label{eqn:lem2:Gamma5}+ \hlf \sum_{i=1}^2 \int_{\Delta_i} y(x,t) (c(x)+c(\xi))K_i(x,\xi)y(\xi,t)d\xi dx.
\end{align}
 Substituting Equations~\eqref{eqn:lem2:Gamma1}~-~\eqref{eqn:lem2:Gamma5} into~\eqref{eqn:lem2:Vdot1} completes the proof.
\end{proof}

\bibliographystyle{plain}
\bibliography{TAC}

\begin{IEEEbiographynophoto} {Aditya Gahlawat}
received the B.Tech degree in mechanical engineering from Punjabi University, Patiala, India in 2007, the M.S. degree in mechanical and aerospace engineering from Illinois Institute of Technology, Chicago, USA in 2009 and the Ph.D. degree in automatique-productique from Universit\'e Grenoble Alpes, St. Martin d'Heres, France, in 2015.
He is currently a Ph.D. candidate in mechanical and aerospace engineering at Illinois Institute of Technology, Chicago, USA
His research focuses on the application of convex optimization based methods for the analysis and control of systems governed by partial differential equations with application to thermonuclear fusion.
Aditya Gahlawat was awarded the Chateaubriand fellowship in 2011 and 2012.
\end{IEEEbiographynophoto}
\begin{IEEEbiographynophoto}{Matthew M. Peet}
received the B.Sc. degree in physics and in aerospace engineering from the University of Texas, Austin, TX, USA, in 1999 and the M.S. and Ph.D. degrees in aeronautics and astronautics from Stanford University in 2001 and 2006, respectively. He was a Postdoctoral Fellow at the National Institute for Research in Computer Science and Control (INRIA), Paris, France, from 2006 to 2008. He was an Assistant Professor of Aerospace Engineering in the Mechanical, Materials, and Aerospace Engineering Department at the Illinois Institute of Technology in Chicago, IL, USA, from 2008 to 2012. Currently, he is an Assistant Professor of Aerospace Engineering in the School for the Engineering of Matter, Transport, and Energy (SEMTE) at Arizona State University, Tempe, AZ, USA, and director of the Cybernetic Systems and Controls Laboratory (CSCL). His research interests are in the role of computation as it is applied to the understanding and control of complex and large-scale systems with an emphasis on methods such as SOS for the optimization of polynomials. Dr. Peet received a National Science Foundation CAREER award in 2011.
\end{IEEEbiographynophoto}

\end{document}